\documentclass[a4paper, autoref, thm-restate, thm-restate, UKenglish]{lipics-v2021}

\usepackage{hyperref}
\usepackage[ruled,vlined]{algorithm2e}
\usepackage{mathtools}
\usepackage{tikz}
\usetikzlibrary{decorations.pathreplacing}
\usepackage{aligned-overset}
\usepackage{xspace}

\newcommand{\Z}{\mathbb{Z}}
\newcommand{\N}{\mathbb{N}}
\newcommand{\R}{\mathbb{R}}

\newcommand{\argmin}{\text{argmin}}

\newcommand{\opt}{\textsc{Opt}^\ast}

\renewcommand{\subset}{\subseteq}
\renewcommand{\epsilon}{\varepsilon}
\newcommand{\rgb}{\{\mathrm{R},\mathrm{G}, \mathrm{B}\}}
\newcommand{\pir}{\pi_{\mathrm{R}}}
\newcommand{\pig}{\pi_{\mathrm{G}}}
\newcommand{\pib}{\pi_{\mathrm{B}}}
\newcommand{\pirg}{\pi_{\mathrm{R}\mathrm{G}}}

\newcommand{\ttsp}{\ensuremath{3\text{-}\mathrm{ETSP}}\xspace}
\newcommand{\ktsp}{\ensuremath{k\text{-}\mathrm{ETSP}}\xspace}
\newcommand{\twotsp}{\ensuremath{2\text{-}\mathrm{ETSP}}\xspace}
\renewcommand{\l}{l} %
\renewcommand{\vec}[1]{\ensuremath{\boldsymbol{#1}}}
\newcommand{\grid}{\mathrm{grid}}
\renewcommand{\P}{\mathrm{Pr}}
\newcommand{\cost}{l}
\renewcommand{\L}{\{0,\dots,L\}^2}
\newcommand{\Lmin}{\{ 0,\dots,L-1\}^2}
\newcommand{\noncrossing}{non-crossing\xspace}

\renewcommand{\t}{^{(t)}}
\newcommand{\bo}{^{(b)}}

\usepackage{environ}
\usepackage{lipsum}
\newif\ifwithappendix

\withappendixfalse

\newcommand{\movetoappendix}[1]{
\ifthenelse{\boolean{withappendix}}{}{#1}
}
\newcommand{\includeinappendix}[1]{
\ifthenelse{\boolean{withappendix}}{#1}{}
}

\pdfoutput=1 %
\hideLIPIcs 
\nolinenumbers

\title{A (5/3+$\epsilon$)-Approximation for \\
Tricolored Non-crossing Euclidean TSP}
\titlerunning{A (5/3+$\epsilon$)-Approximation for tricolored non-crossing Euclidean TSP}

\author{Júlia Baligács}{Technische Universität Darmstadt, Germany}{baligacs@mathematik.tu-darmstadt.de}{https://orcid.org/0000-0003-2654-149X}{}

\author{Yann Disser}{Technische Universität Darmstadt, Germany}{disser@mathematik.tu-darmstadt.de}{https://orcid.org/0000-0002-2085-0454}{}

\author{Andreas Emil Feldmann}{University of Sheffield, UK}{feldmann.a.e@gmail.com}{https://orcid.org/0000-0001-6229-5332}{}

\author{Anna Zych-Pawlewicz\thanks{This work is a part of project BOBR that has received funding from the European Research Council (ERC) under the European Union’s Horizon 2020 research and innovation programme (grant agreement No. 948057).}}{University of Warsaw, Poland}{anka@mimuw.edu.pl}{https://orcid.org/0000-0002-5361-8969}{}

\authorrunning{J. Baligács, Y. Disser, A.E. Feldmann, A. Zych-Pawlewicz}

\Copyright{J. Baligács, Y. Disser, A.E. Feldmann, A. Zych-Pawlewicz} 



\EventShortTitle{ICALP 2024}
\EventAcronym{ICALP}

\begin{document}

\maketitle

\begin{abstract}
In the Tricolored Euclidean Traveling Salesperson problem, we are given~$k=3$ sets of points in the plane and are looking for disjoint tours, each covering one of the sets.
Arora (1998) famously gave a PTAS based on ``patching'' for the case $k=1$ and, recently, Dross et al.~(2023) generalized this result to~$k=2$. 
Our contribution is a $(5/3+\epsilon)$-approximation algorithm for~$k=3$ that further generalizes Arora's approach.
It is believed that patching is generally no longer possible for more than two tours.
We circumvent this issue by either applying a conditional patching scheme for three tours or using an alternative approach based on a weighted solution for $k=2$.

\end{abstract}

\section{Introduction}

We consider the \emph{$k$-Colored Euclidean Traveling Salesperson problem (\ktsp)} where~$k$ sets of points have to be covered by~$k$ disjoint curves in the plane (cf.~Figure~\ref{fig:delta-optimal}).
This is a fundamental problem in geometric network optimization~\cite{DCGHandbook/18} and generalizes the well-known \emph{Euclidean Traveling Salesperson problem (ETSP)}. 
It captures applications ranging from VLSI design~\cite{EN11,KMN01,takahashi1992planeGraphsPaths,takahashi1993rectPaths}
to set visualisation of spatial data~\cite{ARRC11,JGAA-499,EHKP15,HKKLS,reinbacher2008delineating}.

Formally, an instance of \ktsp is a partition $(T_c)_{c \in C}$ of a set of \emph{terminals}~$T \subseteq \mathbb{R}^2$ in the Euclidean plane, where $|C| = k$.
We consider every $c \in C$ to be a \emph{color} and every point in $T_c$ to be of color~$c$.
A solution to the instance is a $k$-tuple $\Pi=(\pi_c)_{c \in C}$ of closed curves in $\mathbb{R}^2$, also referred to as \emph{tours}, such that every curve $\pi_c$ visits all terminals of color $c$, i.e.,~$T_c \subseteq \pi_c$,\footnote{For convenience, we identify curves with their images.} and the curves are pairwise disjoint, i.e., $\pi_c \cap \pi_{c'} = \emptyset$ for $c \neq c'$.

The objective of \ktsp is to minimize the total length of the tours, i.e., to minimize $l(\Pi) := \sum_{c \in C} l(\pi_c)$, where $l(\pi)$ denotes the Euclidean length of~$\pi$. 
It is important to note that, for $k > 1$, no optimum solution may exist (cf.~Figure~\ref{fig:delta-optimal}). 
In order to still define an approximation, we follow the approach of~\cite{dross} by defining the value $\opt := \inf \{\l(\Pi): \Pi \text{ is a solution}\}$ and saying that a solution $\Pi$ is an $\alpha$-approximation if $l(\Pi) \leq \alpha \opt$.

The $\ktsp$ problem inherits NP-hardness from its special case ETSP~\cite{Papadimitriou77} for $k=1$.
It is well-known that 1-ETSP (i.e., ETSP) admits a PTAS~\cite{arora}, and the result was recently extended to a PTAS for \twotsp~\cite{dross}.
The best known approximation factor for~\ttsp was $(10/3 + \varepsilon)$ via doubling of the solution to $3$-Colored Non-crossing Euclidean Steiner Forest from~\cite{bereg}.

\subparagraph*{Our results.}
Our main result is the following.
\begin{theorem}\label{thm:main-thm}
    For every $\epsilon>0$, there is an algorithm that computes a $\left( \frac{5}{3}+\epsilon \right)$-approximation for $\ttsp$ in time $\left( \frac{n}{\epsilon}\right)^{O(1/\epsilon)}$.
\end{theorem}

To prove Theorem~\ref{thm:main-thm}, we adapt Arora's algorithm for Euclidean TSP~\cite{arora}. 
One of the key ingredients of that algorithm is the so-called \emph{Patching Lemma} which allows to locally modify any tour such that the number of crossings with a line segment is bounded, without increasing the length of the tour too much. %
It was shown in~\cite{dross} that this is still possible for two tours,
but it does not seem to be possible for more than two \noncrossing tours (see Figure~\ref{fig:spiral}, \cite{bereg, dross}).
We show how to circumvent this issue by imposing an additional condition on the tours to be patched.
For this, we say that two tours are \emph{$\delta$-close} if they can be connected by a straight line segment of length at most $\delta$ that is disjoint from the third tour (cf.~Figure~\ref{fig:delta-optimal}).

\begin{lemma}\label{lem:intro-patching}
Let $s$ be a straight line segment and $\delta=l(s)$ be its length. Let a solution to $\ttsp$ be given in which two of the three tours are not $\delta$-close.
For every $\delta'>0$, the solution can be modified inside a $\delta'$-neighbourhood of $s$ such that it intersects $s$ in at most $18$ points and its cost is increased by at most $O(\delta)$.
\end{lemma}

\newcommand{\twotoursol}{two-tour presolution\xspace}
\newcommand{\twotoursols}{two-tour presolutions\xspace}
\newcommand{\Twotoursols}{Two-tour presolutions\xspace}
For the case where patching is not possible, %
we take a different approach. For this, we define a \emph{\twotoursol} to be a pair of tours $(\pi_{cc'},\pi_{c''})$ such that $\pi_{cc'}$ visits all terminals colored $c$ and $c'$, and $\pi_{c''}$ visits all terminals colored $c''$.
Such tours can easily be transformed into a feasible solution to $\ttsp$ by ``doubling''~$\pi_{cc'}$~(cf.~Figure~\ref{fig:two-tour-replacement}, Observation~\ref{obs:two-color-solutions}). We call the resulting solution an \emph{induced two-tour solution}. 

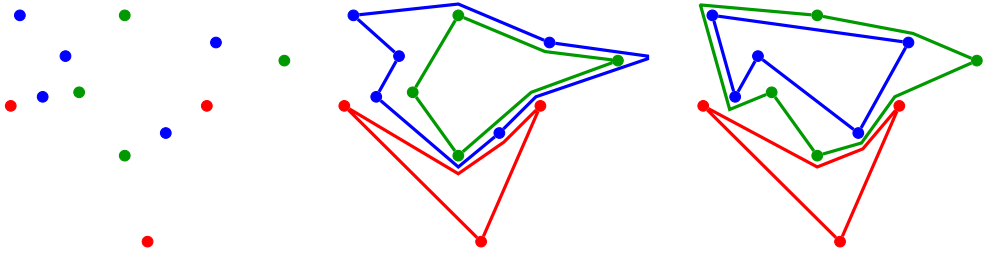
\begin{figure}
\centering
\begin{tikzpicture}[scale=0.6]
\node (A) at (1,-1) [circle, draw, fill, red, inner sep=0pt, minimum size=4pt] {};
\node (B) at (2.3,2) [circle, draw, fill, red, inner sep=0pt, minimum size=4pt] {};
\node (C) at (-2,2) [circle, draw, fill, red, inner sep=0pt, minimum size=4pt] {};
\node (D) at (0.5,0.9) [circle, draw, fill, green!60!black,inner sep=0pt, minimum size=4pt] {};
\node (J) at (-0.5,2.3) [circle, draw, fill, green!60!black,inner sep=0pt, minimum size=4pt] {};
\node (E) at (0.5,4) [circle, draw, fill, green!60!black,inner sep=0pt, minimum size=4pt] {};
\node (F) at (4,3) [circle, draw, fill, green!60!black,inner sep=0pt, minimum size=4pt] {};
\node (G) at (1.4,1.4) [circle, draw, fill, blue,inner sep=0pt, minimum size=4pt] {};
\node (H) at (-1.8,4) [circle, draw, fill, blue,inner sep=0pt, minimum size=4pt] {};
\node (I) at (2.5,3.4) [circle, draw, fill, blue,inner sep=0pt, minimum size=4pt] {};
\node (K) at (-1.3,2.2) [circle, draw, fill, blue,inner sep=0pt, minimum size=4pt] {};
\node (L) at (-0.8,3.1) [circle, draw, fill, blue,inner sep=0pt, minimum size=4pt] {};

\end{tikzpicture}
\hspace{4mm}
\begin{tikzpicture}[scale=0.6]
\node (A) at (1,-1) [circle, draw, fill, red, inner sep=0pt, minimum size=4pt] {};
\node (B) at (2.3,2) [circle, draw, fill, red, inner sep=0pt, minimum size=4pt] {};
\node (C) at (-2,2) [circle, draw, fill, red, inner sep=0pt, minimum size=4pt] {};
\node (D) at (0.5,0.9) [circle, draw, fill, green!60!black,inner sep=0pt, minimum size=4pt] {};
\node (J) at (-0.5,2.3) [circle, draw, fill, green!60!black,inner sep=0pt, minimum size=4pt] {};
\node (E) at (0.5,4) [circle, draw, fill, green!60!black,inner sep=0pt, minimum size=4pt] {};
\node (F) at (4,3) [circle, draw, fill, green!60!black,inner sep=0pt, minimum size=4pt] {};
\node (G) at (1.4,1.4) [circle, draw, fill, blue,inner sep=0pt, minimum size=4pt] {};
\node (H) at (-1.8,4) [circle, draw, fill, blue,inner sep=0pt, minimum size=4pt] {};
\node (I) at (2.5,3.4) [circle, draw, fill, blue,inner sep=0pt, minimum size=4pt] {};
\node (K) at (-1.3,2.2) [circle, draw, fill, blue,inner sep=0pt, minimum size=4pt] {};
\node (L) at (-0.8,3.1) [circle, draw, fill, blue,inner sep=0pt, minimum size=4pt] {};

\draw[very thick, red] (A) to (C) to (0.5,0.5) to (1.5, 1.2) to (B) to (A);
\draw[very thick, green!60!black] (D) to (J) to (E) to (2.4,3.2) to (F) to  (2.1,2.3) to (D);
\draw[very thick, blue] (G) to (0.5, 0.65) to (K) to (L) to (H) to (0.5,4.25) to (I) to (4.65,3.1) to (4.65, 3.04) to (2.2, 2.2) to (G);
\end{tikzpicture}
\hspace{4mm}
\begin{tikzpicture}[scale=0.6]
\node (A) at (1,-1) [circle, draw, fill, red, inner sep=0pt, minimum size=4pt] {};
\node (B) at (2.3,2) [circle, draw, fill, red, inner sep=0pt, minimum size=4pt] {};
\node (C) at (-2,2) [circle, draw, fill, red, inner sep=0pt, minimum size=4pt] {};
\node (D) at (0.5,0.9) [circle, draw, fill, green!60!black,inner sep=0pt, minimum size=4pt] {};
\node (J) at (-0.5,2.3) [circle, draw, fill, green!60!black,inner sep=0pt, minimum size=4pt] {};
\node (E) at (0.5,4) [circle, draw, fill, green!60!black,inner sep=0pt, minimum size=4pt] {};
\node (F) at (4,3) [circle, draw, fill, green!60!black,inner sep=0pt, minimum size=4pt] {};
\node (G) at (1.4,1.4) [circle, draw, fill, blue,inner sep=0pt, minimum size=4pt] {};
\node (H) at (-1.8,4) [circle, draw, fill, blue,inner sep=0pt, minimum size=4pt] {};
\node (I) at (2.5,3.4) [circle, draw, fill, blue,inner sep=0pt, minimum size=4pt] {};
\node (K) at (-1.3,2.2) [circle, draw, fill, blue,inner sep=0pt, minimum size=4pt] {};
\node (L) at (-0.8,3.1) [circle, draw, fill, blue,inner sep=0pt, minimum size=4pt] {};

\draw[very thick, red] (A) to (C) to (0.5,0.65) to (1.5, 1.05) to (B) to (A);
\draw[very thick, green!60!black] (D) to (J) to  (-1.42, 1.92) to (-2.06,4.22) to (-2.05, 4.23) to (E) to (2.6,3.6) to (F) to  (2.2, 2.2) to (1.47, 1.18) to (D);
\draw[very thick, blue] (G) to (L) to (K) to (H) to (I) to (G);
\end{tikzpicture}
\caption{An instance of $\ttsp$ together with two possible solutions. An optimum solution does not exist: The curves can get arbitrarily close but must not touch. Observe that, in the middle subfigure, the red and green tour are not $\delta$-close for any $\delta>0$, as the blue tour lies in between.}
\label{fig:delta-optimal}
\end{figure}

\begin{lemma}\label{lem:intro-two-color-solutions}
For every $\epsilon>0$, there exists $\delta>0$, such that, for every $(1+\epsilon)$-approximate solution to $\ttsp$ in which the two shorter tours are $\delta$-close, we can find a \twotoursol $(\pi_1, \pi_2)$ with $2\l(\pi_1)+\l(\pi_2)\leq \left( \frac{5}{3}+2 \epsilon \right) \cdot \opt$.
\end{lemma}

Similarly as in \cite{arora}, we place a suitable grid on the plane and place so-called \emph{portals} on the grid lines (cf.~Section~\ref{sec:dissection}). Roughly speaking, a solution to $\ktsp$ is \emph{portal-respecting} if it only intersects grid lines at portals and intersects every portal at most a constant number of times (see Section~\ref{sec:dissection} for a formal definition).
Combining the ideas in~\cite{arora} with Lemmas~\ref{lem:intro-patching} and \ref{lem:intro-two-color-solutions}, we obtain the following result.

\begin{theorem}\label{thm:intro-structure-theorem}
For every instance of $\ttsp$ and $\epsilon>0$, either, there is a solution that is a $(1+\epsilon)$-approximation and portal-respecting with respect to a suitable grid, or there is a portal-respecting \twotoursol that induces a $\left( \frac{5}{3}+\epsilon \right)$-approximation.
\end{theorem}

Theorem~\ref{thm:intro-structure-theorem} allows us to restrict ourselves to finding portal-respecting solutions.
The last step is to show that such solutions can be computed in polynomial time using dynamic programming. 
For this, we generalize the approach in \cite{dross} to any number of colors $k$ while simultaneously allowing for weighted tours.
We denote this generalized problem by $\ktsp'$ (see Section~\ref{sec:dynamic-programming} for a formal definition).

\begin{theorem}\label{thm:intro-dynamic-program}
There is a polynomial-time algorithm that computes a parametric solution $\Pi(\lambda)$ to $\ktsp'$ such that $\lim_{\lambda \to 0} \cost\left(\Pi(\lambda)\right)= \opt$.
\end{theorem}

Here, a parametric solution is a function $\Pi \colon (0,\infty) \to \left\{\Pi': \Pi' \text{ is a solution}\right\}$ 
that continuously\footnote{For example, with respect to the Fréchet-distance on the space of curves, which is defined as follows: For $\pi_1,\pi_2 \colon [0,1] \to \R^2$, the Fréchet distance between $\pi_1$ and $\pi_2$ is 
$d_{\mathrm{Fr}}(\pi_1,\pi_2):=\sup_{t \in [0,1]} \lVert \pi_1(t)-\pi_2(t) \rVert$.} interpolates between solutions.
Intuitively, the algorithm of Theorem~\ref{thm:intro-dynamic-program} computes the optimal combinatorial structure of a  solution, i.e., the optimal order in which portals and terminals are visited or bypassed, and the parameter $\lambda$ sets the spacing between the tours.
Importantly, our solution allows to efficiently recover the (non-parametric) solution~$\Pi(\lambda)$ for given~$\lambda > 0$ and to compute $\opt$.

\begin{figure}
\centering
\begin{tikzpicture}[scale=0.2] 
\draw[-, very thick] (-13, 0.5) to (14, 0.5);
\draw[-,very thick] (-13, 0.2) to (-13, 0.8);
\draw[-,very thick] (14, 0.2) to (14, 0.8);
\draw[-, very thick, blue] (-1.07,0) to (0.05,0);
\draw[-,very thick, blue]  (0,0) to (0,1)  to (1,1) to (1,-1) to (-6,-1) to (-6,7) to (7,7) to (7,0);
\draw[-,very thick, blue] (-1,0) to (-1,2) to (2,2) to (2,-2) to (-7,-2) to (-7,8) to (8,8) to (8,0);
\draw[-, very thick, blue] (7,0) to (8,0);
\draw[-, very thick, green!60!black] (-1.95,0) to (-3.07,0);
\draw[-, very thick, green!60!black] (-2,0) to (-2,3) to (3,3) to (3,-3) to (-8,-3) to (-8,9) to (9,9) to (9,0);
\draw[-, very thick, green!60!black] (-3,0) to (-3,4) to (4,4) to (4,-4) to (-9,-4) to (-9,10) to (10,10) to (10,0);
\draw[-, very thick, green!60!black] (9,0) to (10,0);
\draw[-, very thick, red] (-4,0) to (-5,0);
\draw[-,very thick, red]  (-4,0) to (-4,5)  to (5,5) to (5,-5) to (-10,-5) to (-10,11) to (11,11) to (11,0);
\draw[-,very thick, red] (-5,0)  to (-5,6) to (6,6) to (6,-6) to (-11, -6) to (-11,12) to (12,12) to (12,0);
\draw[-, very thick, red] (11,0) to (12,0);
\end{tikzpicture}
\caption{A modified example from \cite{bereg} that is presumably non-patchable.} 
\label{fig:spiral}
\end{figure}
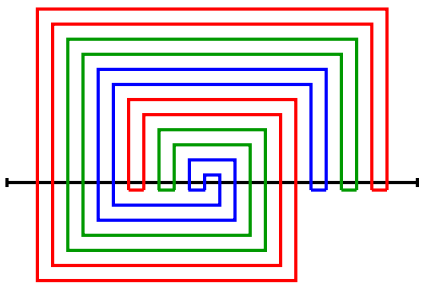

\subparagraph*{Related work.} 

Our results build upon the the celebrated PTAS by Arora~\cite{arora} for ETSP, which was gradually improved~\cite{RaoS98,BartalG13} to an EPTAS~\cite{focs21} with the
running time proven tight under the gap-ETH. 
The $\twotsp$ problem also admits a gap-ETH tight EPTAS~\cite{dross}, which is based on the techniques introduced in~\cite{arora,focs21}, but relies on a more involved patching lemma compared to the one in~\cite{arora}.
The authors of~\cite{dross} claim that patching is unlikely to work for the $\ttsp$ problem (cf.~Figure~\ref{fig:spiral}) and they leave it as a central open problem whether there is a PTAS for $\ttsp$.
We present a $(\frac{5}{3}+\epsilon)$-approximation algorithm that combines a new patching method for three tours with a complementary approach for the case where patching is not possible.
We believe that our approach is applicable for a wider range of non-crossing problems, for instance for the Red-Blue-Green-Yellow Separation problem (cf.~\cite{dross}).

\looseness=-1
Interestingly, similar progress was earlier obtained for the problem
of computing $k$ pairwise non-crossing
Euclidean Steiner trees, one for each color of a
$k$-colored set of terminals in the plane. The $k$-Colored
Non-crossing Euclidean Steiner Forest problem ($k$-CESF for short) was
introduced and studied in~\cite{EHKP15}. Later
Bereg et al.~\cite{bereg}
showed a PTAS for $2$-CESF and a
$(\frac{5}{3}+\epsilon)$-approximation
algorithm for $3$-CESF, leaving the existence of PTAS for $3$-CESF
a main open question. This may suggest that $\frac{5}{3}$ could be some
natural barrier for computing three non-crossing curves interconnecting
three point sets in the plane.
Note that while, similarly to our approach, the algorithm of~\cite{bereg} is based on the technique of Arora~\cite{arora} and a case distinction based on whether patching is possible, our patching procedure is significantly more involved.

Other problems in geometric network optimization include the following: In the \linebreak $k$-Minimum Spanning Tree problem, we have to find a spanning tree connecting a subset of size $k$ of the terminals~\cite{callahan/95, RaoS98}.
In the $k$-Traveling Repairperson problem, we can use $k$ tours (that are allowed to intersect) to cover the terminals, objective to minimizing the latency, i.e., the sum of the times at which a terminal is visited~\cite{chaudhuri/03, chekuri/04, fakcharoenphol/07}.
In the Traveling Salesperson problem with neighbourhoods, the task is to find a shortest tour that visits at least one point in each of a set of neighbourhoods~\cite{gudmundsson/99, mata/95, mitchell/10, safra/06}.

Moreover, the TSP problem has been extensively studied for other metric spaces. For example, it is known that there is a PTAS in the case of a metric space of bounded doubling dimension \cite{bartal/16, gottlieb/15}. On the other hand, it is known that a PTAS for general metric spaces does not exist \cite{marek/15}. Currently, the best approximation algorithm known in general metric spaces is a randomized algorithm by Karlin et al.~\cite{karlin/21}, achieving an approximation factor of~$1.5-10^{-36}$.

\section{\Twotoursols}\label{sec:two-color-solutions}

Recall that a \emph{\twotoursol} for $\ttsp$ is a pair of closed curves $(\pi_{c c'}, \pi_{c''})$ such that~$\pi_{c c'}$ visits all terminals in $T_c \cup T_{c'}$ and $\pi_{c''}$ visits all terminals in $T_{c''}$ for some \linebreak $\{c,c',c''\}=\rgb$, where $\rgb$ denotes throughout the paper the color set $C$ in the case of $\ttsp$, standing for red, green, and blue. 
In this section, we let wlog.~$c''=\mathrm{B}$.
We first investigate how \twotoursols can be transformed into solutions to $\ttsp$.

For this, note that, if we are given a single tour $\pirg$ that visits all red and green terminals, it is possible to replace it by two parametrized disjoint tours $\pir(\lambda)$ and $\pig(\lambda)$ that have Fréchet-distance at most $\lambda$ from $\pirg$ such that $\pir(\lambda), \pig(\lambda)$ visit all red, respectively green, terminals
and  $\lim_{\lambda \to 0} \l(\pir(\lambda))=\lim_{\lambda \to 0} \l(\pig(\lambda))=\l(\pirg)$ (cf.~Figure~\ref{fig:two-tour-replacement}).
Choosing $\lambda>0$ small enough and considering the tours $\pir(\lambda), \pig(\lambda)$, we obtain the following.

\begin{figure}
    \centering
\begin{tikzpicture}[scale=0.45]
\node (A) at (0,0) [circle, draw, fill, red, inner sep=0pt, minimum size=5pt] {};  
\node (B) at (2,1) [circle, draw, fill, green!60!black, inner sep=0pt, minimum size=5pt] {}; 
\node (C) at (4,2) [circle, draw, fill, red, inner sep=0pt, minimum size=5pt] {};
\node (D) at (2,3) [circle, draw, fill, green!60!black, inner sep=0pt, minimum size=5pt] {}; 
\node (E) at (0,4) [circle, draw, fill, red, inner sep=0pt, minimum size=5pt] {};  
\node (F) at (0,2) [circle, draw, fill, green!60!black, inner sep=0pt, minimum size=5pt] {}; 
\draw[-, thick] (0,0) to (4,2) to (0,4) to (0,0);
\phantom{\draw[-, green!60!black, thick, rounded corners=4mm] (2,1) to (4.9,2) to (2,3) to (-0.5, 4.7) to (0,2) to (-0.5, -0.7) to (2,1);}
\end{tikzpicture}
\hspace{5mm}
\begin{tikzpicture}[scale=0.45]
\draw[-,very thick] (2.9,1.69) to (3.04, 1.36);
\node at (3,0.95) {$2\lambda$};
\node (A) at (0,0) [circle, draw, fill, red, inner sep=0pt, minimum size=5pt] {};  
\node (B) at (2,1) [circle, draw, fill, green!60!black, inner sep=0pt, minimum size=5pt] {}; 
\node (C) at (4,2) [circle, draw, fill, red, inner sep=0pt, minimum size=5pt] {};
\node (D) at (2,3) [circle, draw, fill, green!60!black, inner sep=0pt, minimum size=5pt] {}; 
\node (E) at (0,4) [circle, draw, fill, red, inner sep=0pt, minimum size=5pt] {};  
\node (F) at (0,2) [circle, draw, fill, green!60!black, inner sep=0pt, minimum size=5pt] {}; 
\draw[-,thick, red] (A) to (1.9,1.35) to (C) to (1.9,2.65) to (E) to (0.35,2) to (A);
\draw[-, green!60!black, thick] (-0.4,-0.65) to (B) to (4.7,2) to (D) to (-0.4, 4.65) to (F) to (-0.4, -0.65);
\end{tikzpicture}
    \caption{On the left, we have a single curve $\pirg$ visiting all red and green points. On the right, we have replaced $\pirg$ by two parametrized disjoint curves $\pir(\lambda)$ and $\pig(\lambda)$ with  Fréchet-distance at most $\lambda$ to $\pirg$, visiting the terminals of the corresponding color. In particular, the Fréchet-distance between $\pir$ and $\pig$ is at most $2 \lambda$.}
    \label{fig:two-tour-replacement}
\end{figure}
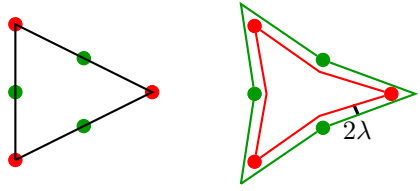

\begin{observation}\label{obs:two-color-solutions}
Fix an instance of $\ttsp$ and let $\pirg, \pib$ be a \twotoursol. For every $\delta>0$, there is a solution to $\ttsp$ of cost at most $2\cdot \l( \pirg)+ \l(\pib)+ \delta$, called an \emph{induced two-tour solution}.
\end{observation}
Next, we show that, if the two shorter tours of a $(1+\epsilon)$-approximation for $\ttsp$ are in some sense close to each other, then there is a good \twotoursol.
However, note that this is not a reduction to $\twotsp$: In $\twotsp$, the objective is to minimize  $\l(\pib)+ \l( \pirg)$. In our case, we need to minimize $\l(\pib)+ 2\cdot \l( \pirg)$, i.e., we need to solve a weighted variant of $\twotsp$.
We will see later that we can compute a $(1+\epsilon)$-approximation for this weighted variant of $\twotsp$ in polynomial time (cf.~Theorem~\ref{thm:dynamic-programming}).

\begin{figure}
\begin{center}
\begin{tikzpicture}[scale=0.85]
\draw[red, thick] plot [smooth cycle] coordinates { (0,1.9) (1,1.9) (2,1.9) (2.5,1.9) (3,1.9) (4,1.9) (4.5, 2.5) (4,3) (2, 3.5) (1, 2.5) (0, 3)};
\draw[blue, thick] plot [smooth] coordinates { (0,1.7) (1,1.7) (2,1.7) (2,1.3) (1,1.4) (0,1.3)};
\draw[green!60!black, thick] plot [smooth cycle] coordinates { (0,1.1) (1,1.1) (2,1.1) (3,1.1) (4,1.1) (4.5, 0) (1,0.5)};
\node at (2.5, 1.9) [circle, draw, fill, inner sep=0.8pt] {};
\node at (2.9, 1.9) [circle, draw, fill, inner sep=0.8pt] {};
\node (x) at (2.7, 1.9) [circle, draw, fill, inner sep=1pt] {};
\node at (2.5, 1.1) [circle, draw, fill, inner sep=0.8pt] {};
\node at (2.9, 1.1) [circle, draw, fill, inner sep=0.8pt] {};
\node (y) at (2.7, 1.1) [circle, draw, fill, inner sep=1pt] {};
\node at (2.7,2.2) {\small{$\vec{x}$}};
\node at (2.3,2.1) {\small{$\vec{x}_1$}};
\node at (3.1,2.1) {\small{$\vec{x}_2$}};
\node at (2.3,0.9) {\small{$\vec{y}_1$}};
\node at (3.1,0.9) {\small{$\vec{y}_2$}};
\node at (2.7,0.8) {\small{$\vec{y}$}};
\draw[-,very thick] (x) to node [right, xshift=-1pt] {$s$} (y);
\end{tikzpicture}
\hspace{3mm}
\begin{tikzpicture}[scale=0.9]
\draw[red, thick] plot [smooth cycle] coordinates { (-0.02,1.88) (1,1.88) (2,1.88) (2.5,1.88) (3,1.88) (4,1.88) (4.53, 2.5) (4,3.02) (1.94, 3.52) (1., 2.54) (-0.03, 3.02)};
\draw[green!60!black, thick] plot [smooth cycle] coordinates {(0.04,1.95) (1,1.95) (2,1.95) (2.5,1.95) (3,1.95) (4,1.95) (4.45, 2.5) (4,2.95) (2, 3.45) (1.02, 2.45) (0.04, 2.93)};
\draw[blue, thick] plot [smooth] coordinates { (0,1.7) (1,1.7) (2,1.7) (2,1.3) (1,1.4) (0,1.3)};
\draw[red, thick] plot [smooth cycle] coordinates { (-0.03,1.1) (1,1.12) (2,1.12) (3,1.12) (4.05,1.12) (4.5, -0.02) (0.92,0.5)};
\draw[green!60!black, thick] plot [smooth cycle] coordinates { (0.08,1.05) (1,1.05) (2,1.05) (3,1.05) (4,1.05) (4.42, 0.03) (1.05,0.55)};
\draw[very thick, white] (2.48, 1.88) to (2.85, 1.88);
\draw[very thick, white] (2.55, 1.95) to (2.75, 1.95);
\draw[very thick, white] (2.5, 1.12) to (2.8, 1.12);
\draw[very thick, white] (2.55, 1.05) to (2.75, 1.05);
\draw[red, thick] (2.48, 1.91) to (2.48, 1.09);
\draw[red, thick] (2.82, 1.91) to (2.82, 1.09);
\draw[green!60!black, thick] (2.55, 1.96) to (2.55, 1.04);
\draw[green!60!black, thick] (2.75, 1.96) to (2.75, 1.04);
\end{tikzpicture}
\end{center}
\vspace{-3mm}
\caption{On the left, we are given a solution to $\ttsp$ where the red and green tour are $\delta$-close. On the right, we see how the solution can be transformed into an induced two-tour solution.} \label{fig:two-color-solutions}
\end{figure}
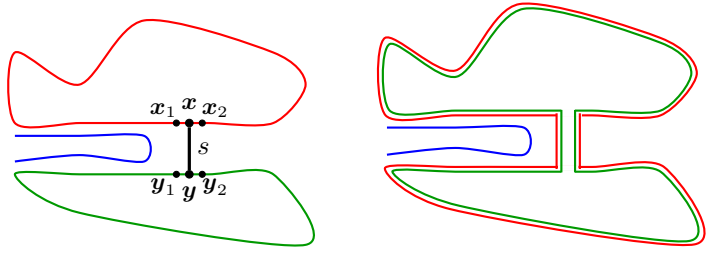

Recall that two tours are $\delta$-close if they can be connected by a straight line segment of length at most $\delta$ that does not intersect the third tour.\includeinappendix{The construction to prove the following Lemma is illustrated in Figure~\ref{fig:two-color-solutions}.\footnote{This and all other missing proofs are deferred to the appendix.}}

\begin{restatable}{lemma}{lemtwotoursol}\label{lem:two-color-solutions}
Let $\Pi=(\pir, \pig, \pib)$ be a solution to a given instance of $\ttsp$ and let $\delta>0$.  Wlog., let $\pib$ be the longest tour, i.e., $\l(\pib)\geq \l(\pir), \l(\pig)$.
Assume that $\pir$ and $\pig$ are $\delta$-close. 
Then, there is a \twotoursol $(\pirg, \pib)$ with 
$2\cdot\l(\pirg)+\l(\pib) \leq \frac{5}{3} \cdot \cost(\Pi) + 8\delta$.
\end{restatable}

\def\prooftwotoursol{
\begin{proof}
The construction is illustrated in Figure~\ref{fig:two-color-solutions}.
Since $\pir$ and $\pig$ are $\delta$-close and terminals are finitely many distinct points, we can pick points $\vec{x}\in \pir$ and $\vec{y}\in \pig$ such that they do not equal any terminal, i.e., $\vec{x},\vec{y}\notin T$, and the straight line segment
$s:=\overline{\vec{x}\vec{y}}$ has length at most $1.5 \cdot \delta$, $s \cap \pib=\emptyset$, $s\cap \pir=\{\vec{x}\}$, and  $s\cap \pir=\{\vec{y}\}$. (Observe that because of the additional condition that the endpoints of $s$ do not lie on terminals, there is not necessarily such a segment of length at most $\delta$.)

In the following, given a closed curve $\pi$ and two points $\vec{x}_1, \vec{x}_2 \in \pi$, by $\pi[\vec{x}_1, \vec{x}_2]$, we denote the shorter of the two subcurves of $\pi$ that connects $\vec{x}_1$ with $\vec{x}_2$ (ties can be broken arbitrarily). 
Observe that it is possible to pick two points $\vec{x}_1, \vec{x}_2 \in \pir$ close enough to $\vec{x}$ and two points $\vec{y}_1, \vec{y}_2 \in \pig$  close enough to $\vec{y}$ such that $\pir[\vec{x}_1,\vec{x}_2]$ and $\pig[\vec{y}_1,\vec{y}_2]$ do not contain any of the terminals, i.e., $T \cap (\pir[\vec{x}_1,\vec{x}_2]\cup \pig[\vec{y}_1,\vec{y}_2])=\emptyset$, and the straight line segments $\overline{\vec{x}_1 \vec{y}_1}, \overline{\vec{x}_2 \vec{y}_2}$ have length at most $2 \cdot \delta$, are nonintersecting, do not intersect $\pib$, and only intersect $\pir, \pig$ in $\vec{x}_1, \vec{x}_2, \vec{y}_1, \vec{y_2}$ (cf.~left side of Figure~\ref{fig:two-color-solutions}).

Therefore, $\pirg:=(\pir \setminus \pir[\vec{x}_1,\vec{x}_2]) \cup \overline{\vec{x}_1\vec{y}_1} \cup (\pig \setminus\pig[\vec{y}_1,\vec{y}_2]) \cup \overline{\vec{x}_2\vec{y}_2}$ forms a closed curve that visits all red and green points and does not intersect $\pib$.
Moreover, we have
\begin{align*}
\l(\pib) + 2 \l(\pirg) &\leq \l(\pib) + 2l(\pir) + 2\l(\pig) + 8 \delta\\
&= \left(\l(\pir)+\l(\pig)+\l(\pib)\right)+ \left(\l(\pir)+\l(\pig)\right) + 8\delta \\
&\leq \frac{5}{3} \cdot \left(\l(\pir)+\l(\pig)+\l(\pib)\right) + 8\delta,
\end{align*}
where we have used in the last inequality that  $\l(\pib)\geq \l(\pir), \l(\pig)$.
\end{proof}
}

\movetoappendix{
\prooftwotoursol}

Note that applying Lemma~\ref{lem:two-color-solutions} to a $(1+\epsilon)$-approximation with $\delta\leq \epsilon \opt/24$
gives a \twotoursol $(\pirg, \pib)$ with 
\begin{equation*}
2\cdot\l(\pirg)+\l(\pib) \leq \frac{5}{3} \cdot (1+\epsilon) \cdot \opt + 8 \delta
= \left( \frac{5}{3}+ \frac{5}{3} \epsilon \right) \cdot \opt +8 \delta \leq  \left( \frac{5}{3}+2 \epsilon \right) \cdot \opt,
\end{equation*}
which completes the proof of Lemma~\ref{lem:intro-two-color-solutions}.

\section{Our structure theorem}

In this section, we prove our structure theorem for $\ttsp$ (cf.~Theorem~\ref{thm:intro-structure-theorem}) which, roughly speaking, states the following: For every $\epsilon>0$, either, there is a \twotoursol that induces a  $\left( \frac{5}{3} + \epsilon \right)$-approximation, or there is a $(1+\epsilon)$-approximate solution that fulfills some additional constraints. 
Later, we will see that it is possible to find a good solution that fulfills these additional constraints and a good \twotoursol in polynomial time.

Our final algorithm for $\ttsp$ will preprocess the input such that the terminals remain distinct points and have integer coordinates. 
This allows us to assume throughout this section that terminals lie in $\L$ for some integer $L$ that is a power of 2. As the problem is not interesting for small $L$, we also assume $L\geq 4$.
In Section~\ref{sec:perturbation}, we explain in more detail how we preprocess the input and show that a near-optimal solution to the preprocessed input can be transformed in polynomial time to a near-optimal solution to the original input.

Following \cite{dross}, we assume without loss of generality that, for every $\epsilon > 0$ and $\delta >0$, there is a $(1+\epsilon)$-approximate solution to $\ttsp$ whose tours consist of straight line segments, where each segment connects two points that each are at distance at most $\delta$ from a terminal. To see this intuitively, interpret each tour of a solution as a sequence of terminals to visit or bypass. The cheapest way to realize such a sequence is by straight line segments with endpoints arbitrarily close to terminals. For this reason, we will assume from now on that all tours that we work with consist of such straight line segments.
Since we assume in this section that terminals lie in $\L$, we have in particular, that the straight line segments have endpoints in $N_\delta \left(\L\right)$, where $N_\delta(A):=\{\vec{x}: \lVert \vec{x}-\vec{a} \rVert < \delta \text{ for some } \vec{a}\in A\}$ denotes the $\delta$-neighbourhood of a set $A$.

\subsection{Dissection and portals}\label{sec:dissection}

In this subsection, we place a suitable grid on the Euclidean plane and place some portals on it through which the tours will later be allowed to cross the grid lines. For this, we follow the definitions as in \cite{arora}. The aforementioned additional constraints for the $(1+\epsilon)$-approximate solution in Theorem~\ref{thm:intro-structure-theorem} strongly relate to this construction and are 
crucial for efficient computation.

Fix an instance of $\ktsp$ with $T \subseteq \L$ where $L$ is a power of two. We pick a \emph{shift vector} $\vec{a}=(a_1, a_2) \in\Lmin$ and consider the square
\begin{equation*}
C(\vec{a}) := \left[-a_1 - \frac{1}{2}, 2L - a_1 - \frac{1}{2} \right] \times \left[-a_2 - \frac{1}{2}, 2L - a_2 - \frac{1}{2} \right],
\end{equation*}
i.e., $C(\vec{a})$ is the square $[0, \dots, 2L]^2$ shifted by $-\vec{a}-(0.5, 0.5)$. 
Note that $C(\vec{a})$ contains every terminal.

The \emph{dissection} $D(\vec{a})$ is a full 4-ary tree defined as follows (illustrated in Figure~\ref{fig:dissection}): Each node is a square in $\R^2$. The root of $D(\vec{a})$ is $C(\vec{a})$. Given a node $S$ of the tree of side length more than one, we partition $S$ into four smaller equal sized squares and these define the four children of $S$. 
If $S$ has side length one, it is a leaf.
Note that this is well-defined because $L$ is a power of two.\footnote{In previous work, the well-known \emph{quad-trees} are defined as a subtree of the dissection, on which the dynamic program of~\cite{arora} is based, which however we do not rely on in this work.}

Given a square $S$, we define its \emph{border edges} to be the unique four straight line segments bounding it and we define its \emph{border} $\partial S$ to be the union of the border edges.

A \emph{grid line} is either a horizontal line containing $(0, -a_2-0.5+k)$ or a vertical line containing $(-a_1-0.5+k, 0)$ for some $k \in \{1, \dots, 2L-1\}$. Note that every border edge of a square in $D(\vec{a})$ is either contained in a grid line or contained in a border edge of $C(\vec{a})$ (which is not on a grid line).
Since terminals have coordinates in $\Z$ and grid lines have coordinates in $0.5+\Z$, no terminal lies on a grid line. More precisely, every terminal lies exactly in the center of a leaf of $D(\vec{a})$.

A \emph{boundary} is a border edge of a non-root node in $D(\vec{a})$ not contained in another border edge (see Figure~\ref{fig:dissection} for an example). 
Observe that its length is $\frac{2L}{2^i}$ for some $i \in \{1, \dots, \log (2L)\}$. Then, we define its \emph{level} as $i$. Note that a grid line only contains boundaries of the same level so we can define the level of a grid line as the level of the boundaries that it contains~(cf.~Figure~\ref{fig:dissection}).

Observe that there is precisely one vertical (respectively horizontal) grid line of level one and, for every $i \in \{1, \dots, \log (2L)-1\}$, there are twice as many grid lines of level $i+1$ as grid lines of level $i$. In total, there are $2L-1$ horizontal and $2L-1$ vertical grid lines. With this, we immediately obtain the following property.

\begin{observation}\label{obs:level}
Let $g$ be either a vertical line containing point $(k-0.5, 0)$ or a horizontal line containing point $(0, k-0.5)$ 
for some $k \in \{1, \dots, L\}$. 
Consider the dissection $D(\vec{a})$ for a vector $\vec{a} \in \{0, \dots, L-1\}^2$ chosen uniformly at random. Then, $g$ is a grid line 
with respect to $D(\vec{a})$, and, for every $i \in \{1, \dots, \log (2L)\}$, we have 
\begin{equation*}
\mathrm{Pr}_{\vec{a}} (\text{the level of $g$ is $i$})=\frac{2^{i-1}}{2L-1}.
\end{equation*}
\end{observation}

The next observation follows immediately from the fact that a square in $D(\vec{a})$ only contains smaller squares of $D(\vec{a})$ so that, in particular, it only contains boundaries of higher levels (cf.~Figure~\ref{fig:dissection}).

\begin{observation}\label{obs:level-intersections}
Let $b=\overline{\vec{x} \vec{y}}$ be a boundary of level $i$. If a grid line $g$ crosses $b^\circ :=b \setminus \{\vec{x}, \vec{y}\}$, the level of $g$ is at least $i+1$.
The levels of the grid lines crossing $b$ through $\vec{x}$ and $\vec{y}$ are at most $i$.
\end{observation}

A \emph{$\delta$-portal} (or, in short, \emph{portal}) on a straight line segment is a subsegment of length $\delta \in (0,1)$. Given a segment~$s$, we define $\grid(s,k, \delta)$ as the set of $k$ equispaced $\delta$-portals on $s$ such that the endpoints of $s$ are contained in the first and last $\delta$-portal respectively.

We place portals on $D(\vec{a})$ as follows: We will choose a large enough integer $r \in \N$ (called the \emph{portal density factor}) and $\delta>0$ (the \emph{portal length}) small enough. Then, for every boundary $b$, we place the portals $\grid(b,r \log L, \delta)$ on $b$ (cf.~Figure~\ref{fig:dissection}). Note hereby that $\log L \in \N$ because $L$ is a power of two.
Observe that, on a higher level boundary, portals are placed more densely, which will turn out to be a key property.

\begin{figure}
\begin{center}
\begin{tikzpicture}[scale=0.8]
\draw[-, thick] (-0.4,0) to node [left] {$2L$} (-0.4,8);
\draw[-, thick] (-0.5,0) to (-0.3,0);
\draw[-, thick] (-0.5,8) to (-0.3,8);
\draw[-, line width=2.7pt] (-0.05,0) to (8,0) to (8,8) to (0,8) to (0,-0.05); 
\draw[-, dashed, thick] (1.5, 0.5) to (5.5, 0.5) to (5.5, 4.5) to (1.5, 4.5) to (1.5, 0.5);
\node at (0,0) [circle, draw, fill, inner sep=0pt, minimum size=2mm] {};
\draw[thick,decorate,decoration={brace,amplitude=5pt}] (0.7, 8.2) -- node [above=2pt, xshift=3pt] {\footnotesize{vertical grid lines}} (7.3, 8.2);
\node at (0, -0.5) {\scriptsize{$\begin{pmatrix}-a_1-0.5\\ -a_2-0.5 \end{pmatrix}$}};
\node at (8.2, 8.3) {$C(\vec{a})$};
\draw[-, line width=1.9pt] (0,4) to (8,4);
\draw[-, line width=1.9pt] (4,0) to (4,8);
\draw[-, line width=2pt, magenta] (4,4) to (4,8);
\node at (8.1, 4) [anchor=west] {{level 1}};
\node at (8.1, 2) [anchor=west] {\footnotesize{level 2}};
\node at (8.1, 6) [anchor=west] {\footnotesize{level 2}};
\foreach \y in {1,3,5,7} {
        \node at (8.1,\y) [anchor=west] {\scriptsize{level 3}};
    }
\draw[-, line width=1.3pt] (2,0) to (2,8);
\draw[-, line width=1.3pt] (6,0) to (6,8);
\draw[-, line width=1.3pt] (0,2) to (8,2);
\draw[-, line width=1.3pt] (0,6) to (8,6);
\draw[-, line width=1.4pt, magenta] (0,6) to (2,6);
\draw[-, very thin] (0,1) to (8,1);
\draw[-, very thin] (0,3) to (8,3);
\draw[-, very thin] (0,5) to (8,5);
\draw[-, very thin] (0,7) to (8,7);
\draw[-, very thin] (1,0) to (1,8);
\draw[-, very thin] (3,0) to (3,8);
\draw[-, very thin] (5,0) to (5,8);
\draw[-, very thin] (7,0) to (7,8);
\draw[-, very thick, magenta] (2,7) to (3,7);
\node at (1.5, 1.5) [circle, draw, fill, red, inner sep=0pt, minimum size=2mm] {};
\node at (4.5, 1.5) [circle, draw, fill, red, inner sep=0pt, minimum size=2mm] {};
\node at (1.5, 0.5) [circle, draw, fill, blue, inner sep=0pt, minimum size=2mm] {};
\node at (2.5, 3.5) [circle, draw, fill, blue, inner sep=0pt, minimum size=2mm] {};
\node at (4.5, 4.5) [circle, draw, fill, green!60!black, inner sep=0pt, minimum size=2mm] {};
\node at (3.5, 2.5) [circle, draw, fill, green!60!black, inner sep=0pt, minimum size=2mm] {};

\foreach \x in {1,3,5,7}{
\foreach \y in {0,0.33,0.67,1,1.33,1.67,2,2.33,2.67,3,3.33,3.67,4,4.33,4.67,5,5.33,5.67,6,6.33,6.67,7,7.33,7.67,8}{
\node at (\x, \y) [circle, draw, fill=white, inner sep=0pt, minimum size=0.71mm] {};
\node at (\y, \x) [circle, draw, fill=white, inner sep=0pt, minimum size=0.71mm] {};
}}

\foreach \x in {2,6}{
\foreach \y in {0,0.67,1.33,2,2.67,3.33,4,4.67,5.33,6,6.67,7.33,8}{
\node at (\x, \y) [circle, draw, fill=white, inner sep=0pt, minimum size=0.86mm] {};
\node at (\y, \x) [circle, draw, fill=white, inner sep=0pt, minimum size=0.86mm] {};
}}

\foreach \y in {0,1.33,2.67,4,5.33,6.67,8}{
\node at (\y,4) [circle, draw, fill=white, inner sep=0pt, minimum size=1.05mm] {};
\node at (4, \y) [circle, draw, fill=white, inner sep=0pt, minimum size=1.05mm] {};
}
\draw[-, very thick, orange] (5.8, 5.85) to (7.5, 5.85) to (7.5, 7.15) to (5.8, 7.15) to (5.8, 5.85);
\end{tikzpicture}
\hspace{4mm}
\begin{tikzpicture}[scale=0.7]
\draw[-, very thick] (-0.9,0) to (4.7,0);
\draw[-] (-0.9,3) to (4.7,3);
\draw[-, very thick] (0,-0.9) to (0,3.9);
\draw[-] (3,-0.9) to (3,3.9);

\node at (0,0) [rectangle, draw, minimum width=1.5mm, minimum height=3mm, inner sep=0pt, anchor=south, fill=gray, fill opacity=0.5] {};
\node at (0,0) [rectangle, draw, minimum width=1.5mm, minimum height=3mm, inner sep=0pt, anchor=north, fill=gray, fill opacity=0.5] {};
\node at (0,0) [rectangle, draw, minimum width=3mm, minimum height=1.5mm, inner sep=0pt, anchor=west, fill=gray, fill opacity=0.5] {};
\node at (0,0) [rectangle, draw, minimum width=3mm, minimum height=1.5mm, inner sep=0pt, anchor=east, fill=gray, fill opacity=0.5] {};

\node at (2,0) [rectangle, draw, minimum width=3mm, minimum height=1.5mm, inner sep=0pt, fill=gray, fill opacity=0.5] {};
\node at (4,0) [rectangle, draw, minimum width=3mm, minimum height=1.5mm, inner sep=0pt, fill=gray, fill opacity=0.5] {};

\node at (3,0.01) [rectangle, draw, minimum width=1.5mm, minimum height=3mm, inner sep=0pt, anchor=south, fill=gray, fill opacity=0.5] {};
\node at (3,-0.01) [rectangle, draw, minimum width=1.5mm, minimum height=3mm, inner sep=0pt, anchor=north, fill=gray, fill opacity=0.5] {};
\node at (3,1.08) [rectangle, draw, minimum width=1.5mm, minimum height=3mm, inner sep=0pt, fill=gray, fill opacity=0.5] {};
\node at (3,1.92) [rectangle, draw, minimum width=1.5mm, minimum height=3mm, inner sep=0pt, fill=gray, fill opacity=0.5] {};
\node at (0,2) [rectangle, draw, minimum width=1.5mm, minimum height=3mm, inner sep=0pt, fill=gray, fill opacity=0.5] {};
\node at (0,3) [rectangle, draw, minimum width=3mm, minimum height=1.5mm, inner sep=0pt, anchor=west, fill=gray, fill opacity=0.5] {};
\node at (0,3) [rectangle, draw, minimum width=3mm, minimum height=1.5mm, inner sep=0pt, anchor=east, fill=gray, fill opacity=0.5] {};

\node at (3,3) [rectangle, draw, minimum width=1.5mm, minimum height=3mm, inner sep=0pt, anchor=south, fill=gray, fill opacity=0.5] {};
\node at (3,3) [rectangle, draw, minimum width=1.5mm, minimum height=3mm, inner sep=0pt, anchor=north, fill=gray, fill opacity=0.5] {};
\node at (3,3) [rectangle, draw, minimum width=3mm, minimum height=1.5mm, inner sep=0pt, anchor=west, fill=gray, fill opacity=0.5] {};
\node at (3,3) [rectangle, draw, minimum width=3mm, minimum height=1.5mm, inner sep=0pt, anchor=east, fill=gray, fill opacity=0.5] {};

\node at (1.08, 3) [rectangle, draw, minimum width=3mm, minimum height=1.5mm, inner sep=0pt, fill=gray, fill opacity=0.5] {};
\node at (1.92,3) [rectangle, draw, minimum width=3mm, minimum height=1.5mm, inner sep=0pt, fill=gray, fill opacity=0.5] {};
\node at (4.08,3) [rectangle, draw, minimum width=3mm, minimum height=1.5mm, inner sep=0pt, fill=gray, fill opacity=0.5] {};

\draw[-, very thick, orange] (-0.7, -0.7) to (4.5, -0.7) to (4.5, 3.7) to (-0.7, 3.7) to (-0.7,-0.7);
\node at (0,-4.5) {};
\end{tikzpicture}
\end{center}
\vspace{-2mm}
\caption{The figure on the left illustrates the dissection $D(\vec{a})$ with $L=4, \vec{a}=(1,0)$. 
The dashed lines denote  $\partial [0,L]^2$.
The three pink lines are examples of boundaries of levels one, two, and three.
The levels of all horizontal grid lines are indicated. 
We have placed 4 portals on every boundary, represented as circles. For better overview, 
we have drawn only one portal at endpoints of boundaries. Note that the endpoint of a boundary is actually contained in up to four portals.
This is illustrated on the right hand side, where the exact portal placement of the marked orange area is given.}\label{fig:dissection}

\end{figure}
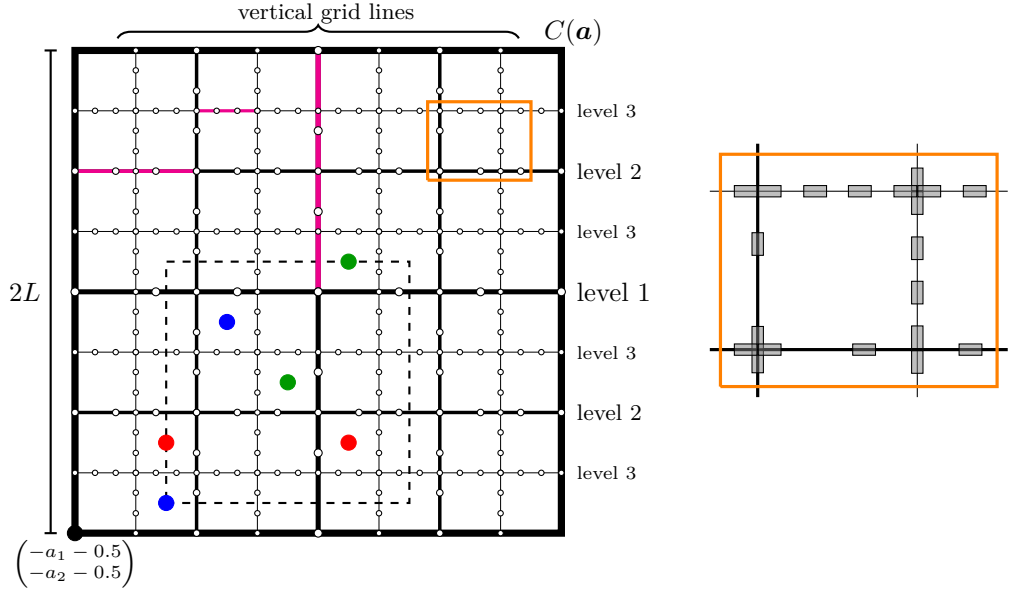

\subsection{Snapping non-crossing curves to portals}
In this subsection, we show that disjoint tours can be modified so 
that they only intersect grid lines in portals, without increasing their lengths too much.
To prove this, we follow the same ideas as in \cite[Section 2.2]{arora}.
Nevertheless, the snapping technique in~\cite[Section 2.2]{arora} needs some adaptation to work in the setting of non-crossing curves. This technique was used in previous work for pairs of non-crossing tours~\cite{dross} and for Steiner trees~\cite{bereg}. Here, we provide a unified framework for this technique, which may be of wider interest and can be applied to a variety of non-crossing Euclidean problems. %

In the following, if $s=\overline{\vec{x}_1 \vec{x}_2}$ is a straight line segment, we let $s^\circ :=s \setminus \{\vec{x}_1, \vec{x}_2\}$. 
This allows us to specify more precisely where the segments are allowed to intersect. In particular, if we require that $s_1^\circ$ and $s_2^\circ$ are disjoint for two segments $s_1$ and~$s_2$, they are allowed to share an endpoint.

\begin{restatable}{lemma}{lemportals}\label{lem:portals}
Let $\mathcal{S}=\{s_1, \dots, s_m\}$ be a finite set of straight line segments in the Eu\-cli\-dean plane such that each $s_i$ connects two points in 
$N_{\frac{1}{4}} \left(\L \right)$
and assume $L\geq 4$.  
Choose a vector \mbox{$\vec{a} \in \Lmin$} uniformly at random
and consider the dissection~$D(\vec{a})$. 
For every portal density factor $r \in \N \setminus \{ 0 \}$, portal length $\delta \in (0,1)$, and $\delta'>0$, there is a set of curves $\mathcal{S}'=\{s_1', \dots, s_m'\}$ (not necessarily straight line segments) such that 
\begin{enumerate}[a)]
\item $s_i'$ differs from $s_i$ only in $N_{\delta'}(G)$ where $G$ denotes the union of all grid lines in $D(\vec{a})$,
\item if the segments $s_i^\circ$ are pairwise disjoint, then the curves $(s_i')^\circ$ are pairwise disjoint as well,
\item every $s_i'$ intersects every boundary $b$ of $D(\vec{a})$ only in the portals $\grid(b, r \log L, \delta)$, i.e., $s_i' \cap b \subseteq \grid(b, r \log L, \delta)$,
\item no $s_i'$ contains an intersection point of two grid lines, i.e., $g_1 \cap g_2 \cap s_i'=\emptyset$ for every $i \in \{1, \dots, m\}$ and grid lines $g_1 \neq g_2$,
\item the curves of $\mathcal{S}'$ intersect the grid lines in finitely many points, 
\item $\mathbb{E}_{\vec{a}}\left[\l(\mathcal{S}')-\l(\mathcal{S})\right] \leq 7\sqrt{2} \cdot \frac{l(S)}{r}$, where $\l(\mathcal{S}):=\sum_{s\in\mathcal{S}} \l(s)$.
\end{enumerate}
\end{restatable}

\begin{figure}
\begin{center}
\begin{tikzpicture}[scale=1.3]
\draw[-] (0,0) to (3.4,0);
\draw[-, line width=3pt] (0.07,0) to (0.33,0);
\draw[-, line width=3pt] (1.07,0) to (1.33,0);
\draw[-, line width=3pt] (2.07,0) to (2.33,0);
\draw[-, line width=3pt] (3.07,0) to (3.33,0);
\draw[, thick, red] (1.3,0.5) to (2.1,-0.5);
\draw[, thick, blue] (1.7,0.5) to (2.2,-0.5);
\draw[->, thick] (4,0) to (4.5,0);
\node at (5,0) {};
\draw[-] (1.2,0.7) to (1.2, -0.7);
\draw[, thick, green!60!black] (0.9,0.5) to (1.5,-0.5);
\end{tikzpicture}
\begin{tikzpicture}[scale=1.3]
\filldraw[fill=orange!20, opacity=1,fill opacity=0.45, draw=orange!20] (0,-0.2) rectangle (3.4, 0.2);
\draw[-] (0,0) to (3.4,0);
\draw[-, line width=3pt] (0.07,0) to (0.33,0);
\draw[-, line width=3pt] (1.07,0) to (1.33,0);
\draw[-, line width=3pt] (2.07,0) to (2.33,0);
\draw[-, line width=3pt] (3.07,0) to (3.33,0);
\draw[, thick, red] (1.3,0.5) to (1.62,0.1);
\draw[, thick, red] (1.78,-0.1) to (2.1,-0.5);
\draw[, thick, red] (1.62,0.1) to (2.15,0.1);
\draw[, thick, red] (1.78,-0.1) to (2.15,-0.1);
\draw[, thick, red] (2.15,0.1) to (2.15,-0.1);
\draw[, thick, blue] (1.7,0.5) to (1.875,0.15);
\draw[, thick, blue] (2.025,-0.15) to (2.2,-0.5);
\draw[, thick, blue] (2.25,0.15) to (1.875,0.15);
\draw[, thick, blue] (2.25,-0.15) to (2.025,-0.15);
\draw[, thick, blue] (2.25,0.15) to (2.25,-0.15);
\draw[-] (1.2,0.7) to (1.2, -0.7);
\draw[, thick, green!60!black] (0.9,0.5) to (1.14,0.1);
\draw[, thick, green!60!black] (1.29,-0.15) to (1.5,-0.5);
\draw[, thick, green!60!black] (1.14,0.1) to (1.29, 0.1) to (1.29, -0.15);
\draw[decorate,decoration={brace,amplitude=3pt}] (-0.1,0.01) -- (-0.1,0.2) node[midway,left, xshift=-2pt] {$\mu$};
\draw[decorate,decoration={brace,amplitude=3pt}] (-0.1,-0.2) -- (-0.1,-0.01) node[midway,left, xshift=-2pt] {$\mu$};
\end{tikzpicture}
\end{center}
\caption{On the left, the red and blue tour cross a boundary outside of a portal and the green tour crosses in an intersection of grid lines. On the right, the tours are modified in a $\mu$-Neighbourhood of the grid line such that they are still non-crossing, only cross boundaries at portals and do not cross in intersections of grid lines.}\label{fig:portals}
\end{figure}
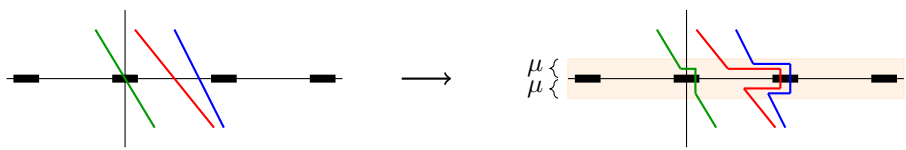

\def\prooflemportals{
\begin{proof}
Throughout this proof, the segments of $\mathcal{S}$ are gradually modified to obtain the set of curves $\mathcal{S'}$ in the end. All the modifications satisfy the invariant, that the current set of curves only intersects grid lines in finitely many points that we refer to as  \emph{crossings}. %
We use a parameter $\mu>0$ throughout the proof, that will be carefully set later on in the proof, in a way that the condition \mbox{$\mu<\min \{\delta, \delta', 0.25\}$} is satisfied. We will modify the segments inside $N_{\mu}(G)$. This implies then condition~a) and that the segments are not modified in $N_{\frac{1}{4}} \left(\L\right)$, i.e., the segments are not modified in some neighbourhood around their endpoints.
This makes it a bit easier to prove part d) later because we can assume than that the segments are disjoint in the considered area where we make modifications.

We obtain the set $\mathcal{S}'$ from $\mathcal{S}$ as follows:
Consider every boundary $b$ of $D(\vec{a})$ one by one (in arbitrary order) and apply the following modifications: (Type 1) Move every crossing on $b$ with $\mathcal{S}$ to the nearest portal on $b$ as illustrated by the red and blue segments in Figure~\ref{fig:portals}. (Type 2) If a crossing on $b$ with $\mathcal{S}$ lies precisely on the intersection of two grid lines, move it inside the portal by at most $\mu$ as illustrated by the green segment in  Figure~\ref{fig:portals}.

Observe that we only modify the tours in an $N_{\mu}(b)$ and that the modifications can be done such that we do not create additional crossings between the segments. Therefore, it is immediate that conditions a), b), and e) are fulfilled after this procedure.

To prove parts c) and d), we have to show that, once a boundary has been considered, we neither create additional crossings on this boundary outside of portals nor do we move existing crossings out of portals.
Note that, when we apply a modification of Type 1 on a boundary $b$, the new lines created for redirecting a crossing to a portal can create new crossings with grid lines perpendicular to $b^\circ$. We argue that these new crossings lie in a portal:
We have seen in Observation~\ref{obs:level-intersections} that these lines are of higher levels. In particular, if $b'$ is a boundary perpendicular to $b^\circ$, $b \cap b'$ is an endpoint of $b'$ and, therefore, lies in a portal of $b'$. The new crossings created on $b'$ are at distance at most $\mu<\delta$ from $b$, in particular from an endpoint of $b'$, so that they are placed in a portal of $b'$. 

In contrast, a modification of Type 2 on $b$ can relocate a crossing on a boundary $b'$ perpendicular to $b$ (not only $b^\circ$) but it does not increase the total number of crossings on $b'$ (cf.~green curve in Figure~\ref{fig:portals}). We show that a boundary on which a crossing is relocated, has not been considered yet: 
Let $\vec{p}$ be a crossing on $b$ that needs to be moved by a modification of Type 2. Then,  $\vec{p}$ lies in the intersection of two or more (at most four) boundaries. Let $b'$ be the last boundary considered before $b$ that contains $\vec{p}=b \cap b'$.
By construction of the modification procedure, after $b'$ was considered, there were no crossings lying in the intersection of $b'$ with any grid line. In particular, there was no crossing at $b \cap b'$. Since $b'$ was the last boundary considered before $b$ that contains $b \cap b'$, no crossing can be moved to $b \cap b'$ in the meantime. This gives a contradiction so that we obtain that no other boundary containing the crossing~$\vec{p}$ was considered yet.
To summarize, when a boundary $b$ is considered, modifications of Type 1 and Type 2 do not create additional crossings outside of portals or relocate crossings on boundaries that were already considered.
With this, we obtain that conditions c) and d) are fulfilled as well.

It is left to show that condition f) is satisfied.
First, note that the cost of a modification of Type 2 is at most $\mu$. Since we apply finitely many modifications of Type 2 (at most $L^2$), $\mu$ can be chosen small enough such that the total cost of modifications of Type 2 is at most $\l(\mathcal{S})/r$. 
Next, as already explained, the new crossings created by a modification of Type 1 already lie in portals so that they do not need to be moved. Therefore, the total number of modifications of Type 1 we need to apply is at most the number of crossings of the unmodified set $\mathcal{S}$ with grid lines, i.e., $|G \cap \mathcal{S}|$ where $G \cap \mathcal{S}:=\bigcup_{i=1}^m G \cap s_i$.

Consider some grid line $g$. Let $i$ denote its level in $D(\vec{a})$. 
Then, the length of a boundary lying on $g$ is $\frac{2L}{2^i}$ so that the distance between two portals on $g$ is at most $\frac{2L}{2^i (r \log L -1)}$.
Hence, the cost of moving a crossing on $g$ to the closest portal is also at most $\frac{2L}{2^i (r \log L -1)}$. 
Overall, the total cost of modifications of Type 1 on boundaries on $g$ is in expectation at most
\begin{align*}
&\sum\limits_{i=1}^{\log 2L} \P(g \text{ has level } i) \cdot |g \cap \mathcal{S}| \cdot \frac{2L}{2^i (r \log L -1)}\\
\overset{\text{Obs.~\ref{obs:level}}}&{=} 
\sum\limits_{i=1}^{\log 2L} \frac{2^{i-1}}{2L-1} \cdot |g \cap \mathcal{S}| \cdot \frac{2L}{2^i (r \log L -1)}\\
&= \frac{|g \cap \mathcal{S}|}{r} \cdot \frac{1}{2} \cdot \underbrace{\frac{2L}{2L-1}}_{\leq 2}  \cdot \underbrace{\frac{\log 2L}{(\log L -\frac{1}{r})}}_{\leq 3}
\leq 3 \cdot \frac{|g \cap \mathcal{S}|}{r},
\end{align*}
where we have used in the last inequality that $L\geq 4$.
Summing up over all grid lines gives the following estimate on the total cost of modifications of Type 1
\begin{equation}\label{eq:number-of-crossings}
\sum\limits_{g \text{ is a grid line}} 3\cdot \frac{|g \cap \mathcal{S}|}{r} = 3\cdot \frac{|G \cap \mathcal{S}|}{r}.
\end{equation}
As the last step, we relate the number of crossings $|\mathcal{S} \cap G|$ to the total length of $\mathcal{S}$. 
For this, consider a straight line segment $s_i \in \mathcal{S}$ between two points in $N_{\frac{1}{4}} \left(\L\right)$. Let $(x_1, y_1),(x_2, y_2) \in \L$ be the closest points to the endpoints of $s_i$ and consider the straight line segment $s_i'':=\overline{(x_1, y_1)(x_2, y_2)}$. Note that $|G \cap s_i''|=|G \cap s_i|$.
If $(x_1, y_1)=(x_2, y_2)$, we have $G \cap s_i= G \cap s_i''=\emptyset$, so assume that $(x_1, y_1)\neq(x_2, y_2)$ and, in particular, $\l(s_i'')\geq 1$. 
By choice of $(x_1,y_1), (x_2,y_2)$, we have that the distance between the endpoints of $s$ and $s'$ is at most 1/4 so that we obtain with triangle inequality
\begin{equation*}
\l(s_i'')\leq \l(s_i)+2 \cdot \frac{1}{4} \overset{\l(s'')\geq 1}{\leq} \l(s_i)+\frac{1}{2}\cdot\l(s_i'')
\end{equation*}
so that $\l(s_i'') \leq 2 \l(s_i)$. Therefore, it suffices to relate the number of crossings of $s_i''$ with grid lines to $\l(s_i'')$.

Note that the length of $s_i''$ is $\sqrt{|x_1-x_2|^2+|y_1-y_2|^2}$. Moreover, $s_i''$ crosses precisely $|x_1-x_2|$ vertical grid lines and $|y_1-y_2|$ horizontal grid lines. Hence,
\begin{align*}
|s_i'' \cap G|^2&=\left(|x_1-x_2|+|y_1-y_2| \right)^2
= 2 \left( |x_1-x_2|^2+|y_1-y_2|^2 \right) - \left( |x_1-x_2|-|y_1-y_2| \right)^2 \\
&\leq 2 \cdot \left( |x_1-x_2|^2+|y_1-y_2|^2 \right)
= 2 \cdot \l(s_i'')^2,
\end{align*}
so that $|s_i \cap G|=|s_i'' \cap G| \leq \sqrt{2} \cdot \l(s_i'')\leq 2\sqrt{2} \cdot \l(s_i)$. Combining this with \eqref{eq:number-of-crossings}, the total cost of modifications of Type~1 is at most
\begin{equation}\label{eq:total-length}
3 \cdot \frac{|G \cap \mathcal{S}|}{r}
=\frac{3}{r}\sum\limits_{i=1}^m |s_i \cap G|
\leq \frac{3}{r}\sum\limits_{i=1}^m 2\sqrt{2} \cdot \l(s_i)
= 6 \sqrt{2} \cdot \frac{\l(\mathcal{S})}{r}.
\end{equation}
Recall that we have seen that the total cost of modifications of Type 2 is at most $\l(\mathcal{S})/r$.
Together with \eqref{eq:total-length}, we obtain 
\begin{equation*}
\mathbb{E}_{\vec{a}}\left[\l(S')-\l(S)\right]
\leq 6 \sqrt{2} \cdot \frac{\l(\mathcal{S})}{r} +  \frac{\l(\mathcal{S})}{r} \leq 7 \sqrt{2} \cdot \frac{\l(\mathcal{S})}{r}.\qedhere
\end{equation*}
\end{proof}
}

\movetoappendix{\prooflemportals }
\includeinappendix{
\begin{proof}[Proof sketch]
We apply the modifications illustrated in Figure~\ref{fig:portals} one by one on every boundary. When estimating the cost of these modifications, we make use of the fact that, on a boundary of a higher level, the portals are placed more densely and, by Observation~\ref{obs:level}, given a fixed line intersected by $\mathcal{S}$, the probability that it is of a lower level in $D(\vec{a})$ is small. 
The total cost of the modifications in Figure~\ref{fig:portals} depends on the total number of intersection points with grid lines. Therefore, the last step is to relate the number of intersection points to $\l(\cal S)$, making use of the fact that $\cal S$ consists of straight lines segments connecting points close to $\L$.
\end{proof}
}

\subsection{The patching technique for three disjoint tours}
In the previous section, we have seen that a (reasonable) solution to $\ktsp$ can be modified such that it only intersects grid lines in portals. In this section, we investigate how the tours can further be modified to reduce the number of intersection points per portal.
This will be important for our algorithm because it considers all possible ways that a solution can cross the squares of $D(\vec{a})$ through the portals. To obtain a reasonable running time, we need a constant bound on the number of crossings.
As briefly explained in the introduction, we cannot hope to show this for every solution, even for $k=3$. Therefore, we restrict ourselves to $\ttsp$ and show the desired properties for this problem under some additional assumptions.

Before delving into the proof, let us introduce some useful notation and establish the prerequisites.
Let  $\pir, \pig, \pib$ be a solution to an instance of $\ttsp$ and $s$ be a straight line segment such that $\pi_c \cap s$ consists of finitely many distinct points for all $c \in \rgb$. Then we say that $s$ is \emph{non-aligned} to the solution and we call the points in $\bigcup_{c\in C} \pi_c \cap s$ \emph{crossings}.
We define an order on the crossings by rotating the plane such that $s$ is parallel to the $x$-axis and then ordering them by their $x$-coordinates. This allows us to speak of a crossing to be ``next to'' or ``in between'' other ones.
The \emph{color} of a crossing $\vec{x}$, denoted~$c(\vec{x})$, is $d$ if $\vec{x}\subseteq s \cap \pi_d$.
With this, we can classify the occurring patterns by sequences of the three colors and we call this a \emph{crossing pattern}. For example, if $x_1, \dots, x_6$ denote the ordered crossings, by the crossing pattern $RRGGBB$, we mean that $c(x_1)=c(x_2)=R$, $c(x_3)=c(x_4)=G$ and $c(x_5)=c(x_6)=B$.

Our work builds on existing results for one and two tours.
Arora~\cite{arora} showed that the number of crossings of a single tour with a line segment can be reduced as follows~(cf.~Figure~\ref{fig:patching-arora}).

\begin{lemma}[Arora~\cite{arora}]\label{lem:patching-one-tour}
    Let $\pi$ be a closed curve and $s$ be a non-aligned straight line segment. For every $\delta>0$, there is a curve $\pi'$ differing from $\pi$ only inside $N_\delta(s)$ such that $|s \cap \pi' | \leq 2$ and $\l(\pi')\leq \l(\pi) + 3 \cdot \l(s).$
\end{lemma}

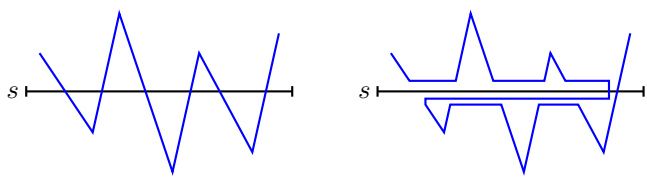
\begin{figure}
\centering
\begin{tikzpicture}[scale=0.35]
\draw[thick,-] (-0.5,0.05) to (9.5,0.05);
\draw[thick] (-0.5, -0.15) to (-0.5, 0.25);
\draw[thick] (9.5, -0.15) to (9.5, 0.25);
\node at (-1,0) {$s$};
\draw[thick, blue] (0,1.5) to (2,-1.5) to (3,3) to (5,-3) to (6,1.5) to (8,-2.25) to (9,2.25);
\end{tikzpicture}
\hspace{5mm}
\begin{tikzpicture}[scale=0.35]
\draw[thick] (-0.5, -0.15) to (-0.5, 0.25);
\draw[thick] (9.5, -0.15) to (9.5, 0.25);
\draw[thick,-] (-0.5,0.05) to (9.5,0.05);
\draw[thick, blue] (0,1.5) to (0.7,0.45) to (2.44,0.45) to (3,3) to (3.85, 0.45) to (5.77,0.45) to (6,1.5) to (6.56,0.45) to (8.2, 0.45) to (8.2, -0.225) to (1.3, -0.225) to (1.3,-0.45) to (2,-1.5) to (2.233,-0.45) to (4.15,-0.45) to (5,-3) to (5.567, -0.45) to (7.04, -0.45) to (8,-2.25) to (9,2.25);
\node at (-1,0) {$s$};
\end{tikzpicture}
\caption{Illustration of the patching scheme in \cite{arora} for a single tour: On the left, we are given a tour $\pi$ intersecting the segment $s$ in six points. On the right, we see how the tour can be modified such that the number of crossings is at most two and the length of is increased at most by $3 \cdot \l(s)$.}
\label{fig:patching-arora}
\end{figure}

\begin{figure}
\begin{center}
\begin{tikzpicture}[scale=0.55]
\draw[thick] (0,0) to (7,0);
\draw[thick] (0,0.1) to (0,-0.1);
\draw[thick] (7,0.1) to (7,-0.1);

\draw[very thick, blue] (1,2) to (1,1);
\draw[very thick, blue] (1,-2) to (1,-1);
\draw[very thick, blue] (1.5,-1) to (1.5,-2);
\draw[very thick, blue] (1.5,1) to (1.5,2);

\draw[very thick, red] (2.5,1) to (2.5,2);
\draw[very thick, red] (3,1) to (3,2);
\draw[very thick, red] (2.5,-1) to (2.5,-2);
\draw[very thick, red] (3,-1) to (3,-2);

\draw[very thick, blue] (4,-2) to (4,-1);
\draw[very thick, blue] (4.5,-2) to (4.5,-1);
\draw[very thick, blue] (4,1) to (4,2);
\draw[very thick, blue] (4.5,1) to (4.5,2);

\draw[very thick, red] (5.5,1) to (5.5,2);
\draw[very thick, red] (6,1) to (6,2);
\draw[very thick, red] (5.5,-2) to (5.5,-1);
\draw[very thick, red] (6,-2) to (6,-1);

\draw[very thick, blue] (1,1) to (1,-1); 
\draw[very thick, blue, bend right=20] (1.5,1) to (4,1);
\draw[very thick, blue] (1.5,-1) to (1.5,-0.7) to (4.5,0.7) to (4.5,1);
\draw[very thick, blue] (4,-1) to (4.5,-1); 

\draw[very thick, red] (2.5, 1) to (3,1);
\draw[very thick, red, bend left=20] (3,-1) to (5.5,-1);
\draw[very thick, red] (6, -1) to (6,1);
\draw[very thick, red] (2.5,-1) to (2.5,-0.7) to (5.5,0.7) to (5.5,1);

\foreach \x in {1,...,4}{
\node at (1.5*\x-0.5,1) [circle, draw, fill, inner sep=1pt] {};
\node at (1.5*\x-0.5,-1) [circle, draw, fill, inner sep=1pt] {};
\node at (1.5*\x,1) [circle, draw, fill, inner sep=1pt] {};
\node at (1.5*\x,-1) [circle, draw, fill, inner sep=1pt] {};
}

\draw[blue, very thick, dashed, bend left=20] (4.25, -2) to (3.5, -2.4);
\draw[red, very thick, dashed, bend left=20] (2.75, 2) to (3.5, 2.4);
\draw[blue, very thick, dashed, dash pattern=on 2.2pt off 2pt] (4,2.1) to [out=90, in=180] (4.25,2.4) to [out=0, in=90] (4.5,2.1);
\draw[red, very thick, dashed, dash pattern=on 2.2pt off 2pt] (2.5,-2.1) to [out=-90, in=180] (2.75,-2.4) to [out=0, in=-90] (3,-2.1);
\end{tikzpicture}
\hspace{4mm}
\begin{tikzpicture}[scale=0.55]
\draw[thick] (0.5,0) to (7,0);
\draw[thick] (0.5,0.1) to (0.5,-0.1);
\draw[thick] (7,0.1) to (7,-0.1);

\draw[very thick, blue] (1.5,-1) to (1.5,-2);
\draw[very thick, blue] (1.5,1) to (1.5,2);

\draw[very thick, red] (2.5,1) to (2.5,2);
\draw[very thick, red] (3,1) to (3,2);
\draw[very thick, red] (2.5,-1) to (2.5,-2);
\draw[very thick, red] (3,-1) to (3,-2);

\draw[very thick, blue] (4,-2) to (4,-1);
\draw[very thick, blue] (4.5,-2) to (4.5,-1);
\draw[very thick, blue] (4,1) to (4,2);
\draw[very thick, blue] (4.5,1) to (4.5,2);

\draw[very thick, red] (5.5,1) to (5.5,2);
\draw[very thick, red] (6,1) to (6,2);
\draw[very thick, red] (5.5,-2) to (5.5,-1);
\draw[very thick, red] (6,-2) to (6,-1);

\draw[very thick, blue, bend right=20] (1.5,1) to (4,1);
\draw[very thick, blue] (1.5,-1) to (1.5,-0.7) to (4.5,0.7) to (4.5,1);
\draw[very thick, blue] (4,-1) to (4.5,-1); 

\draw[very thick, red] (2.5, 1) to (3,1);
\draw[very thick, red, bend left=20] (3,-1) to (5.5,-1);
\draw[very thick, red] (6, -1) to (6,1);
\draw[very thick, red] (2.5,-1) to (2.5,-0.7) to (5.5,0.7) to (5.5,1);

\node at (1.5,1) [circle, draw, fill, inner sep=1pt] {};
\node at (1.5,-1) [circle, draw, fill, inner sep=1pt] {};
\foreach \x in {2,3,4}{
\node at (1.5*\x-0.5,1) [circle, draw, fill, inner sep=1pt] {};
\node at (1.5*\x-0.5,-1) [circle, draw, fill, inner sep=1pt] {};
\node at (1.5*\x,1) [circle, draw, fill, inner sep=1pt] {};
\node at (1.5*\x,-1) [circle, draw, fill, inner sep=1pt] {};
}

\draw[blue, very thick, dashed, bend left=20] (4.25, -2) to (3.5, -2.4);
\draw[red, very thick, dashed, bend left=20] (2.75, 2) to (3.5, 2.4);
\draw[blue, very thick, dashed, dash pattern=on 2.2pt off 2pt] (4,2.1) to [out=90, in=180] (4.25,2.4) to [out=0, in=90] (4.5,2.1);
\draw[red, very thick, dashed, dash pattern=on 2.2pt off 2pt] (2.5,-2.1) to [out=-90, in=180] (2.75,-2.4) to [out=0, in=-90] (3,-2.1);
\end{tikzpicture}
\end{center}
\caption{Illustration of the patching schemes for the crossing patterns  $BBRRBBRR$ (on the left) and $BBRRBBRR$ (on the right), see also~\cite{dross}. Observe that the connections illustrated by dashed lines must exist (up to symmetry). %
}\label{fig:patching-two-tours}
\end{figure}
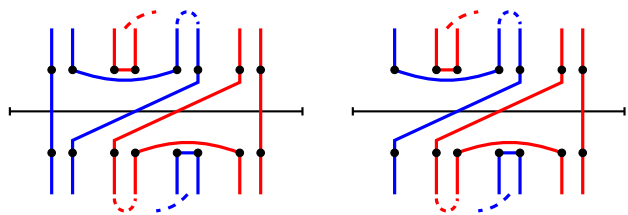

Dross et al.~\cite{dross} proved that the number of crossings between two disjoint tours and a straight line segment $s$ can be reduced to a constant number, at additional cost $O(\l(s))$. 
Here, we only need the two-color patching schemes for the two special crossing patterns given in the following result.
To see why the following lemma holds, one can carefully investigate the proof in \cite{dross} or observe that the scheme illustrated in Figure~\ref{fig:patching-two-tours} works as desired.

\begin{lemma}[Dross et al.~\cite{dross}]\label{lem:patching-two-tours}
    Let $\pi_c, \pi_d$ be disjoint closed curves and $s$ be a non-aligned straight line segment. Assume the crossing pattern is given by $ccddccdd$ or $cddccd$. For every $\delta>0$, there are disjoint closed curves $\pi_c', \pi_d'$ differing from $\pi_c$ and $\pi_d$ only inside $N_\delta(s)$ such that the new crossing pattern is $ccdd$, respectively $cdd$, and $\l(\pi_c')+\l(\pi_d')\leq \l(\pi_c)+\l(\pi_d) + 4 \cdot \l(S).$
\end{lemma}

Now, we have all the prerequisites in place to give a patching procedure in the case that no red crossing is next to a green crossing on the considered segment (or for any other other choice of two colors), i.e., to prove Lemma~\ref{lem:intro-patching}. For this, we first give a more precise formulation of Lemma~\ref{lem:intro-patching}.

\begin{restatable}[Tricolored Patching]{lemma}{patching}\label{lem:patching}
Let  $\pir, \pig, \pib$ be disjoint closed curves and $s$ be a non-aligned straight line segment.
Assume that, in the crossing pattern on $s$, there is no red crossing next to a green crossing.
Then, for every $\delta>0$, there are disjoint closed curves $\pir', \pib', \pig'$ such that 
\begin{enumerate}[a)]
\item for all $c \in \rgb$, $\pi_c$ differs from $\pi_c'$ only inside $N_\delta(s)$,
\item $|(s \cap \pir') \cup (s \cap \pib') \cup (s \cap \pig')| \leq 18$,
\item $\l(\pir')+ \l(\pig')+\l(\pib')\leq \l(\pir)+\l(\pig)+\l(\pib)+79 \cdot \l(s)$.
\end{enumerate}
\end{restatable}

\def\proofpatching{
\begin{proof}
We modify the three curves in several steps to reduce the number of crossings with~$s$.
Hereby, we arrange the crossings into \emph{groups}, where a group is a maximal set of consecutive monochromatic crossings. The \emph{color of a group $\cal G$}, denoted $c, (\cal G)$%
is then the color of the crossings in this group.

The road map for our proof is roughly the following: First, we reduce the number of crossings inside a group by applying Lemma~\ref{lem:patching-one-tour}. Next, we further simplify the occurring patterns by applying~Lemma~\ref{lem:patching-two-tours}. Then, we give a new patching scheme for the remaining occurring patterns for three colors.

\newcommand{\ovpi}{\overline{\pi}}
\paragraph*{Reducing the number of crossings per group.} For every group $\cal G$ that contains more than two crossings, we apply the following modifications: Let $s'(\cal G)$ be the subsegment of $s$ beginning in the first crossing of $\cal G$ and ending in the last crossing of $\cal G$. Apply Lemma~\ref{lem:patching-one-tour} to the tour of color $c(\cal G)$ and the subsegment $s'(\cal G)$.
After this, we obtain three tours $\overline{\pi}\mathrm{R}, \overline{\pi}_{\mathrm{G}}, \overline{\pi}_{\mathrm{B}}$ differing from $\pir, \pib, \pig$ only in a $\delta$-neighbourhood of $s$ such that each group contains at most two crossings and the total cost of the modifications is at most $\sum_{G}3\cdot \l(s'(\mathcal{G})) \leq 3 \cdot \l(s)$.
Therefore,  
\begin{equation}\label{eq:crossings-per-group}
\l(\overline{\pi}_{\mathrm{R}})+\l(\overline{\pi}_{\mathcal{G}})+\l(\overline{\pi}_{\mathrm{B}})\leq \l(\pir) + \l(\pi_\mathcal{G}) + \l(\pib) + 3 \cdot \l(s).
\end{equation}

\paragraph*{Bounding the number of red and green groups with only one crossing.}
After the previous step, every group contains now either one or two crossings. 
Our next goal is to bound the number of red and green groups with only one crossing. 

For this, note that, for every $c,c' \in \rgb, c \neq c'$ and every choice of two crossings $\vec{x},\vec{y}$ of color $c$, the number of crossings of color $c'$ in between them is even: 
The union of the subcurves $\ovpi_c[\vec{x},\vec{y}] \cup s[\vec{x},\vec{y}]$ together form a closed curve, where  $\ovpi_c[\vec{x},\vec{y}]$ denotes one of the two subcurves of $\ovpi_c$ that connect $\vec{x}$ and $\vec{y}$. Since $\ovpi_{c'}$ is also a closed curve, the number of intersection points in  $(\ovpi_c[\vec{x},\vec{y}] \cup s[\vec{x},\vec{y}]) \cap \ovpi_{c'}$ must be even. Because $\ovpi_c$ and $\ovpi_{c'}$  are disjoint, the number of crossings in  $\ovpi_{c'}\cap s[\vec{x},\vec{y}]$ must be even.

Next, note that, if we are given a red or green group $\cal G$, it can only have blue groups as neighbours because we assumed that no red crossing is next to a green crossing. 
If it is neighboured by two blue groups, the number of crossings in $\cal G$ must be even, i.e., two, because otherwise there is an odd number of crossings of color $c(\cal G)$ in between these two blue groups.
Therefore, if $\cal G$ has only one crossing, it has only one neighbouring group, which means that it is either the first or the last group on the segment $s$. In particular, there are at most two red or green groups consisting of only one crossing.

\paragraph*{Patching bichromatic patterns.}
In the next step, we eliminate alternating sequences of red and blue, respectively green and blue groups. As already explained, a bichromatic red-blue or green-blue crossing pattern cannot contain $RBR$ or $GBG$, so it suffices to be able to patch the patterns $BBRRBBRR$, $BRRBBRR$, $BBGGBBGG$, $BGGBBGG$.
For this, consider an arbitrary red-blue (or green-blue) bichromatic pattern and apply the following modifications. While the crossing pattern contains one the above patterns ($BBRRBBRR$, $BRRBBRR$, $BBGGBBGG$ or $BGGBBGG$), pick the leftmost starting such pattern. Let $s'$ be the shortest subsegment containing the chosen pattern (i.e., beginning in the first crossing and ending in the last crossing of the chosen pattern) and apply Lemma~\ref{lem:patching-two-tours} to $s'$. 
The resulting pattern is then $BBRR$, respectively $BRR$, $BBGG$, $BGG$, and the cost of this step is at most $4\l(s')$.
Additionally, move the at most four resulting crossings sufficiently close to the rightmost endpoint of $s'$, which gives an additional cost of at most $8\l(s')$. 
The cost of modifications for a single choice of the pattern is then at most $12 \cdot \l(s')$. 

Let $\vec{x}$ be the leftmost crossing in $s'$ after the modifications and $\vec{y}$ be the right endpoint of $s'$, i.e., $\vec{x}$ is sufficiently close to $\vec{y}$. 
Note that, after the modifications, none of the four possible patterns can start left of $\vec{x}$ and none of the patterns can entirely lie left of~$\vec{y}$. This is because the chosen pattern was leftmost.
Therefore, if a new subsegment~$s_1'$ is chosen in the next step, it starts right of $\vec{x}$ and ends right of $\vec{y}$.
In particular, when applying the procedure on $s_1'$, the crossings can be moved close enough to the right endpoint of $s_1'$ so that they do not lie in $s'$. Then, the subsegment~$s_2'$ chosen after $s_1'$ does not intersect $s'$.
This means that the sequence $s_i'$ of the $s'$ chosen in the modifications above is such that each point on $s$ is contained in at most two of the $s_i'$.
Therefore, the total cost of the modifications is at most $\sum_i 12 \cdot \l(s_i') \leq 24 \cdot \l(s)$.

With this, we obtain three curves $\tilde{\pi}_{\mathrm{R}}, \tilde{\pi}_{\mathrm{G}}, \tilde{\pi}_{\mathrm{B}}$ differing from $\pir, \pig, \pib$ only inside $N_\delta(s)$ and having total length at most
\begin{equation}\label{eq:bichromatic-patterns}
\l(\tilde{\pi}_{\mathrm{R}})+ \l(\tilde{\pi}_{\mathrm{G}}) + \l(\tilde{\pi}_{\mathrm{B}}) \leq 
\l(\ovpi_{\mathrm{R}})+\l(\ovpi_{\mathrm{G}})+\l(\ovpi_{\mathrm{B}})+24 \cdot \l(S)
\overset{\eqref{eq:crossings-per-group}}{\leq} 
\l(\pir) + \l(\pig) + \l(\pib) + 27 \cdot \l(S)
\end{equation}
such that the crossing pattern fulfills the following:
\begin{enumerate}[i)]
\item no red group is next to a green group,
\item every group contains at most two crossings,
\item every red or green group which is not the first or last group along $s$ contains precisely two crossings,
\item the subpatterns $BBRRBBRR$, $BRRBBRR$, $BBGGBBGG$ and $BGGBBGG$ are not contained,
\item between any two crossings of a color $c$, the number of crossings of color $c'\neq c$ is even, in particular, the patterns $RRBRR$ and $GGBGG$ are not contained.
\end{enumerate}
To sum up, the possible crossing patterns are as follows: First, i) and ii) imply that the sequence alternates between a blue group and a red or green group where every group has one or two crossings. 
Second, combining iii) with iv) and v), we obtain that 
there is no subsequence of blue group, red group, blue group, red group, blue group. And similarly, there is no subsequence of blue group, green group, blue group, green group, blue group.

\renewcommand{\b}{B^\ast}
Therefore, the crossing pattern of $\tilde{\pi}_{\mathrm{R}}, \tilde{\pi}_{\mathrm{G}}, \tilde{\pi}_{\mathrm{B}}$ with $s$ has the following form: $\mathcal{G}_1 P \mathcal{G}_2$ where 
$\mathcal{G}_1,\mathcal{G}_2 \in \{ \emptyset, R,G,RR,GG\}$, and $P$ is a subsequence of $\b RR\b GG\b RR\b GG \dots$, where each $\b$ can be replaced independently by $B$ or $BB$.
The next step will be to modify the sequence $P$ such that it has bounded length.

\paragraph*{Patching trichromatic patterns}
In this step, we show that a pattern $\b RR\b GG\b RR\b GG \dots$ can be patched in a way that the resulting pattern contains either at most one green or at most one red group.
For this, assume that we are given a pattern where we have at least two groups of each color. Then, it contains the subpattern $GG\b RR\b GG\b RR$, where the roles of $R$ and $G$ can be exchanged. We number the groups in that pattern by $G_1, B_1, R_1, B_2, G_2, B_3, R_2$ (cf.~Figure~\ref{fig:pattern}).
Our goal is to reduce this pattern to $GG\b RR$. 

As a first step, note that, in this pattern, every $\b$ needs to be replaced by the same choice in $\{B,BB\}$: If $B_1$ and $B_2$ do not have the same number of crossings, the number of blue crossings between $G_1$ and $G_2$ is odd, which gives a contradiction. Similarly, $B_2$ and $B_3$ contain the same number of crossings.
Therefore, we investigate the two patterns $GGBBRRBBGGBBRR$ and $GGBRRBGGBRR$.

Even though we make adjustments only in a $\delta$-neighbourhood of $s$, it is important to study how the crossings are connected via the entire tours because we need to ensure that our patching procedure does not disconnect a tour.

For this, we say that the top of a crossing $\vec{x}$ on $s$ is \emph{connected to the top} of another crossing $\vec{y}$ of the same color $c$ if, when traveling along the curve $\pi_c$ from $\vec{x}$ in the direction upwards from $s$, the first other crossing in $s \cap \pi_c$ encountered is $\vec{y}$ and we enter $\vec{y}$ from the top. This is similarly defined for \emph{bottoms} of crossings and combinations of tops and bottoms of crossings. Moreover, we say that the top of a group $\cal G$ is connected to the top of another group $\mathcal{G}'$ if one of the tops of the crossings in $\cal G$ is connected to one of the tops of the crossings in $\cal G'$.
We denote the tops and bottoms of the groups by $G_1\t , G_1\bo, \dots$ (cf.~Figure~\ref{fig:pattern}).

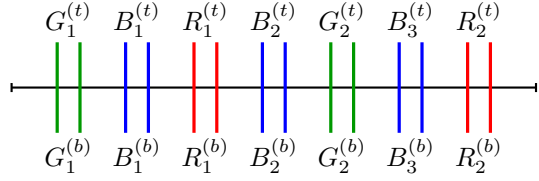
\begin{figure}
\begin{center}
\begin{tikzpicture}[scale=0.6]
\draw[thick] (0,0) to (11.5,0);
\draw[thick] (0,0.1) to (0,-0.1);
\draw[thick] (11.5,0.1) to (11.5,-0.1);

\draw[very thick, green!60!black] (1,1) to (1,-1);
\draw[very thick, green!60!black] (1.5,1) to (1.5,-1);

\draw[very thick, blue] (2.5,1) to (2.5,-1);
\draw[very thick, blue] (3,1) to (3,-1);

\draw[very thick, red] (4,1) to (4,-1);
\draw[very thick, red] (4.5,1) to (4.5,-1);

\draw[very thick, blue] (5.5,1) to (5.5,-1);
\draw[very thick, blue] (6,1) to (6,-1);

\draw[very thick, green!60!black] (7,1) to (7,-1);
\draw[very thick, green!60!black] (7.5,1) to (7.5,-1);

\draw[very thick, blue] (8.5,1) to (8.5,-1);
\draw[very thick, blue] (9,1) to (9,-1);

\draw[very thick, red] (10,1) to (10,-1);
\draw[very thick, red] (10.5,1) to (10.5,-1);

\node at (1.25, 1.5) {$G_1\t$};
\node at (2.75, 1.5) {$B_1\t$};
\node at (4.25, 1.5) {$R_1\t$};
\node at (5.75, 1.5) {$B_2\t$};
\node at (7.25, 1.5) {$G_2\t$};
\node at (8.75, 1.5) {$B_3\t$};
\node at (10.25, 1.5) {$R_2\t$};

\node at (1.25, -1.5) {$G_1\bo$};
\node at (2.75, -1.5) {$B_1\bo$};
\node at (4.25, -1.5) {$R_1\bo$};
\node at (5.75, -1.5) {$B_2\bo$};
\node at (7.25, -1.5) {$G_2\bo$};
\node at (8.75, -1.5) {$B_3\bo$};
\node at (10.25, -1.5) {$R_2\bo$};
\end{tikzpicture}
\end{center}
\caption{The crossing pattern $GGBBRRBBGGBBRR$ in a $\delta$-neighbourhood of $s$.}\label{fig:pattern}
\end{figure}

We start by investigating the connections of the middle blue group $B_2$. First, note that there cannot be a connection from $B_2$ to any blue group on the opposite side, i.e., from $B_2\t$ to $B_i\bo$ or from $B_2\bo$ to $B_i\t$ for any $i\in\{1,2,3\}$: To see this, observe that, in Figure~\ref{fig:impossible-connection}~(a), 
including any of the dotted lines would lead to the green, respectively red, groups being disconnected.
Therefore, $B_2$ can only be connected to a blue group on the same side, i.e., $B_2\t$ can be connected to $B_i\t$ and $B_2\bo$ can be connected to $B_j\bo$ for some $i,j \in \{1,2,3\}$.

Next, note that we cannot have hereby the connections with $i=j=1$ or $i=j=3$: As illustrated in Figure~\ref{fig:impossible-connection}~(b), this would lead to the red (in case $i=j=1$), respectively green (in case $i=j=3$), groups to be disconnected.

After eliminating impossible connections, we show that the connections with $\{i,j\}=\{1,3\}$ exist. First, note that it is not possible that both, $B_2\t$ and $B_2\bo$, are only connected to themselves, i.e., $i=j=2$, because then $B_2$ would not be connected to any other blue group. Assume wlog.~that the connection between $B_2\t$ and $B_1\t$ exists and we want to show that the connection between $B_2\bo$ and $B_3\bo$ exists as well. 
Observe in Figure~\ref{fig:impossible-connection}~(c) that the existence of the connection between $B_2\t$ and $B_1\t$ implies the existence of a connection between $R_1\bo$ and $R_3$, i.e., one of the red dashed connections must exist (note that the red connection cannot bypass $s$ from the left because this would lead to $G_1$ being disconnected from $G_2$).
With this, observe that any of the dotted blue connections would lead to $G_2$ being disconnected from $G_1$.
There must be a connection from $B_3$ to some other blue group and we can now see from Figure~\ref{fig:impossible-connection}~(c) that the only remaining possibility is a connection between $B_2\bo$ and $B_3\bo$.

To summarize, we have seen that there is a connection between $B_2\t$ and $B_i\t$, and between $B_2\bo$ and $B_j\bo$ for some $\{i,j\}=\{1,3\}$ and there are no other connections of $B_2$. In particular, given the case that each blue group consists of two crossings, the tops and bottoms of both crossings in $B_2$ are connected to the same other blue group.
Assume from now on wlog.~that $i=1$ and $j=3$.
The two possible crossing patterns are illustrated in Figure~\ref{fig:possible-connections} together with the connections of $B_2$. By the dashed red and green line, we indicate that $R_1\bo$ must be connected to $R_2$ and $G_2\t$ must be connected to $G_1$.

\begin{figure}
\begin{center}
\begin{tikzpicture}[scale=0.3]
\node at (5.25,-5) {\footnotesize(a)};
\node at (0,-4.5) {};
\node at (0,2) {};
\node at (0,-3.5) {};
\draw[thick] (0,0) to (11.5,0);
\draw[thick] (0,0.1) to (0,-0.1);
\draw[thick] (11.5,0.1) to (11.5,-0.1);

\draw[very thick, green!60!black] (1,1) to (1,-1);
\draw[very thick, green!60!black] (1.5,1) to (1.5,-1);

\draw[very thick, blue] (2.5,1) to (2.5,-1);
\draw[very thick, blue] (3,1) to (3,-1);

\draw[very thick, red] (4,1) to (4,-1);
\draw[very thick, red] (4.5,1) to (4.5,-1);

\draw[very thick, blue] (5.5,1) to (5.5,-1);
\draw[very thick, blue] (6,1) to (6,-1);

\draw[very thick, green!60!black] (7,1) to (7,-1);
\draw[very thick, green!60!black] (7.5,1) to (7.5,-1);

\draw[very thick, blue] (8.5,1) to (8.5,-1);
\draw[very thick, blue] (9,1) to (9,-1);

\draw[very thick, red] (10,1) to (10,-1);
\draw[very thick, red] (10.5,1) to (10.5,-1);

\draw[dotted, blue, very thick] (5.75, 1) to [out=110, in=90] (-1.1,0);
\draw[dotted, blue, very thick] (-1.1,0) to [out=-90, in=-120] (5.75,-1);
\draw[dotted, blue, very thick] (5.75, 1) to [out=140, in=80] (-0.5,0);
\draw[dotted, blue, very thick] (-0.5,0) to [out=-90, in=-120] (2.75,-1);
\draw[dotted, blue, very thick] (5.8, 1) to [out=90, in=0] (2.1,3) to [out=180, in=90]  (-1.7,0) to [out=-92, in =-110] (8.9,-1);
\end{tikzpicture}
\hspace{1mm}
\begin{tikzpicture}[scale=0.3]
\node at (5.25,-5) {\footnotesize(b)};
\node at (0,2.7) {};
\node at (0,-4.5) {};
\node at (0,2) {};
\node at (0,-3.5) {};
\draw[thick] (0,0) to (11.5,0);
\draw[thick] (0,0.1) to (0,-0.1);
\draw[thick] (11.5,0.1) to (11.5,-0.1);

\draw[very thick, green!60!black] (1,1) to (1,-1);
\draw[very thick, green!60!black] (1.5,1) to (1.5,-1);

\draw[very thick, blue] (2.5,1) to (2.5,-1);
\draw[very thick, blue] (3,1) to (3,-1);

\draw[very thick, red] (4,1) to (4,-1);
\draw[very thick, red] (4.5,1) to (4.5,-1);

\draw[very thick, blue] (5.5,1) to (5.5,-1);
\draw[very thick, blue] (6,1) to (6,-1);

\draw[very thick, green!60!black] (7,1) to (7,-1);
\draw[very thick, green!60!black] (7.5,1) to (7.5,-1);

\draw[very thick, blue] (8.5,1) to (8.5,-1);
\draw[very thick, blue] (9,1) to (9,-1);

\draw[very thick, red] (10,1) to (10,-1);
\draw[very thick, red] (10.5,1) to (10.5,-1);

\draw[dotted, blue, bend left=80, very thick] (2.75, 1) to (5.7, 1);
\draw[dotted, blue, bend right=80, very thick] (2.75, -1) to (5.7, -1);
\end{tikzpicture}
\hspace{1mm}
\begin{tikzpicture}[scale=0.3]
\node at (5.25,-5) {\footnotesize (c)};
\node at (0,2.7) {};
\node at (0,-4.5) {};
\node at (0,2) {};
\node at (0,-3.5) {};
\draw[thick] (0,0) to (11.5,0);
\draw[thick] (0,0.1) to (0,-0.1);
\draw[thick] (11.5,0.1) to (11.5,-0.1);

\draw[very thick, green!60!black] (1,1) to (1,-1);
\draw[very thick, green!60!black] (1.5,1) to (1.5,-1);

\draw[very thick, blue] (2.5,1) to (2.5,-1);
\draw[very thick, blue] (3,1) to (3,-1);

\draw[very thick, red] (4,1) to (4,-1);
\draw[very thick, red] (4.5,1) to (4.5,-1);

\draw[very thick, blue] (5.5,1) to (5.5,-1);
\draw[very thick, blue] (6,1) to (6,-1);

\draw[very thick, green!60!black] (7,1) to (7,-1);
\draw[very thick, green!60!black] (7.5,1) to (7.5,-1);

\draw[very thick, blue] (8.5,1) to (8.5,-1);
\draw[very thick, blue] (9,1) to (9,-1);

\draw[very thick, red] (10,1) to (10,-1);
\draw[very thick, red] (10.5,1) to (10.5,-1);

\draw[dashed, blue, bend left=60, very thick] (2.75, 1) to (5.7, 1);
\draw[dotted, blue, bend left=60, very thick] (5.8, 1) to (8.75, 1);
\draw[dotted, blue, bend left=80, very thick] (2.73, 1) to (8.77, 1);
\draw[dashed, red, bend right=30, very thick] (4.25, -1) to (10.25, -1);
\draw[dashed,red,very thick] (4.25,-1.1) to [out=-60, in=-90] (11.8,0) to [out=90, in=-30] (11.2,1.4)  to [out=150,in=90] (10.25,1);
\end{tikzpicture}

\caption{Impossible connections of the blue groups.}\label{fig:impossible-connection}
\end{center}
\end{figure}
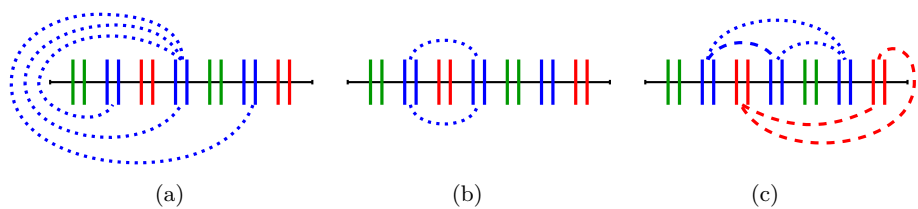

From this, one can observe that cutting the tours open at points at distance at most~$\delta$ from $s$ and reconnecting as illustrated in Figure~\ref{fig:tricolorpatching}, the three curves are still connected and therefore, still form closed \noncrossing curves. Note that the new intersection pattern is $GG\b RR$ as desired.
One can also see in the figure that every line included for the reconnection has length at most $\l(s)$. Since each line connects two points where the tours were cut open and there were 14, respectively 11, crossings, the total cost of the reconnection is at most $14 \cdot \l(s)$.
Moreover, the modified tours have at most 6 crossings so that moving them sufficiently close to the right endpoint of $s$ gives an additional cost of at most $12 \cdot \l(s)$.
Overall, this shows that we can reduce the intersection pattern $GG\b RR\b GG\b RR$ to $GG\b RR$ and move the crossings sufficiently close to the right endpoint of $s$ at cost $(14+12)\cdot \l(s)=26 \cdot \l(s)$.

Given this, we proceed similarly as in the paragraph on bichromatic patterns to reduce the number of crossings with $s$: While the intersection pattern contains the subsequence $GG\b RR\b GG\b RR$ (or with exchanged roles of $G$ and $R$), choose the leftmost such pattern and apply the described patching scheme to the shortest subsegment $s'$ of $s$ containing the subpattern and move the new crossings sufficiently close to the right endpoint of $s$. 
We have seen that the cost of a single application is $26\cdot \l(s')$ and, similarly as for bichromatic patterns, one can show that every point in $s$ is contained in at most two of the chosen $s'$. Therefore, the total cost of the modifications is at most $52 \cdot \l(s)$. Let $\pir', \pig', \pib'$ denote the resulting tours. Then, we have
\begin{equation}\label{eq:trichromatic-modification}
\l(\pir')+\l(\pig')+\l(\pib') \leq 
\l(\tilde{\pi}_{\mathrm{R}})+ \l(\tilde{\pi}_{\mathrm{G}}) + \l(\tilde{\pi}_{\mathrm{B}}) +52 \cdot \l(S) 
\overset{\eqref{eq:bichromatic-patterns}}{\leq} 
\l(\pir) + \l(\pig) + \l(\pib) + 79 \cdot \l(S).
\end{equation}

To summarize, we have shown that the intersection pattern of the modified tours $\pir', \pig', \pib'$ with $S$ is of the form $\mathcal{G}_1 \mathcal{P} \mathcal{G}_2$ where $\mathcal{G}_1,\mathcal{G}_2 \in \{ \emptyset, R,G,RR,GG\}$ and $\mathcal{P}$ is a subsequence of $\b RR\b GG\b RR\b GG \dots$ that does not contain two red and two green groups. A longest possible such subsequence is $BBRRBBGGBBRRBB$. Therefore, the lengths of $\mathcal{G}_1$ and $\mathcal{G}_2$ are at most two and the length of $\mathcal{P}$ is at most $14$. This implies that the total number of crossings is at most $18$.

Moreover, observe that we have only modified the tours inside $N_\delta(s)$ and that by equation~\eqref{eq:trichromatic-modification} the total length of $\pir', \pig', \pib'$ is as desired.
\end{proof}
}

\movetoappendix{\proofpatching}

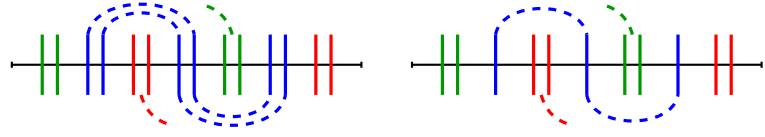
\begin{figure}
\begin{center}
\begin{tikzpicture}[scale=0.4]
\node at (0,2.7) {};
\node at (0,-3) {};
\draw[thick] (0,0) to (11.5,0);
\draw[thick] (0,0.1) to (0,-0.1);
\draw[thick] (11.5,0.1) to (11.5,-0.1);

\draw[very thick, green!60!black] (1,1) to (1,-1);
\draw[very thick, green!60!black] (1.5,1) to (1.5,-1);

\draw[very thick, blue] (2.5,1) to (2.5,-1);
\draw[very thick, blue] (3,1) to (3,-1);

\draw[very thick, red] (4,1) to (4,-1);
\draw[very thick, red] (4.5,1) to (4.5,-1);

\draw[very thick, blue] (5.5,1) to (5.5,-1);
\draw[very thick, blue] (6,1) to (6,-1);

\draw[very thick, green!60!black] (7,1) to (7,-1);
\draw[very thick, green!60!black] (7.5,1) to (7.5,-1);

\draw[very thick, blue] (8.5,1) to (8.5,-1);
\draw[very thick, blue] (9,1) to (9,-1);

\draw[very thick, red] (10,1) to (10,-1);
\draw[very thick, red] (10.5,1) to (10.5,-1);

\draw[dashed, blue, bend left=80, very thick] (3, 1) to (5.5, 1);
\draw[dashed, blue, bend left=80, very thick] (2.5, 1) to (6, 1);
\draw[dashed, blue, bend right=80, very thick] (6, -1) to (8.5, -1);
\draw[dashed, blue, bend right=80, very thick] (5.5, -1) to (9, -1);
\draw[dashed, red, very thick, bend right=30] (4.25,-1) to (5.25, -2);
\draw[dashed, green!60!black, very thick, bend right=30] (7.25,1) to (6.25, 2);
\end{tikzpicture}
\hspace{3mm}
\begin{tikzpicture}[scale=0.4]
\node at (0,2.7) {};
\node at (0,-3) {};
\draw[thick] (0,0) to (11.5,0);
\draw[thick] (0,0.1) to (0,-0.1);
\draw[thick] (11.5,0.1) to (11.5,-0.1);

\draw[very thick, green!60!black] (1,1) to (1,-1);
\draw[very thick, green!60!black] (1.5,1) to (1.5,-1);

\draw[very thick, blue] (2.75,1) to (2.75,-1);

\draw[very thick, red] (4,1) to (4,-1);
\draw[very thick, red] (4.5,1) to (4.5,-1);

\draw[very thick, blue] (5.75,1) to (5.75,-1);

\draw[very thick, green!60!black] (7,1) to (7,-1);
\draw[very thick, green!60!black] (7.5,1) to (7.5,-1);

\draw[very thick, blue] (8.75,1) to (8.75,-1);

\draw[very thick, red] (10,1) to (10,-1);
\draw[very thick, red] (10.5,1) to (10.5,-1);

\draw[dashed, blue, bend left=80, very thick] (2.75, 1) to (5.75, 1);
\draw[dashed, blue, bend right=80, very thick] (5.75, -1) to (8.75, -1);
\draw[dashed, red, very thick, bend right=30] (4.25,-1) to (5.25, -2);
\draw[dashed, green!60!black, very thick, bend right=30] (7.25,1) to (6.25, 2);
\end{tikzpicture}
\end{center}
\vspace{-7mm}
    \caption{We show that, in the two considered crossing patterns, the connections illustrated by dashed lines must exist. 
    }
    \label{fig:possible-connections}
\end{figure}

\begin{figure}
\begin{center}
\begin{tikzpicture}[scale=0.55]
\draw[thick] (0,0) to (11.5,0);
\draw[thick] (0,0.1) to (0,-0.1);
\draw[thick] (11.5,0.1) to (11.5,-0.1);

\draw[very thick, green!60!black] (1,2) to (1,1);
\draw[very thick, green!60!black] (1,-2) to (1,-1);
\draw[very thick, green!60!black] (1.5,-1) to (1.5,-2);
\draw[very thick, green!60!black] (1.5,1) to (1.5,2);

\draw[very thick, blue] (2.5,1) to (2.5,2);
\draw[very thick, blue] (3,1) to (3,2);
\draw[very thick, blue] (2.5,-1) to (2.5,-2);
\draw[very thick, blue] (3,-1) to (3,-2);

\draw[very thick, red] (4,-2) to (4,-1);
\draw[very thick, red] (4.5,-2) to (4.5,-1);
\draw[very thick, red] (4,1) to (4,2);
\draw[very thick, red] (4.5,1) to (4.5,2);

\draw[very thick, blue] (5.5,1) to (5.5,2);
\draw[very thick, blue] (6,1) to (6,2);
\draw[very thick, blue] (5.5,-2) to (5.5,-1);
\draw[very thick, blue] (6,-2) to (6,-1);

\draw[very thick, green!60!black] (7,1) to (7,2);
\draw[very thick, green!60!black] (7.5,1) to (7.5,2);
\draw[very thick, green!60!black] (7,-2) to (7,-1);
\draw[very thick, green!60!black] (7.5,-2) to (7.5,-1);

\draw[very thick, blue] (8.5,1) to (8.5,2);
\draw[very thick, blue] (9,1) to (9,2);
\draw[very thick, blue] (8.5,-2) to (8.5,-1);
\draw[very thick, blue] (9,-2) to (9,-1);

\draw[very thick, red] (10,1) to (10,2);
\draw[very thick, red] (10.5,1) to (10.5,2);
\draw[very thick, red] (10,-2) to (10,-1);
\draw[very thick, red] (10.5,-2) to (10.5,-1);

\draw[dashed, blue, bend left=80, very thick] (3, 2) to (5.5, 2);
\draw[dashed, blue, bend left=80, very thick] (2.5, 2) to (6, 2);
\draw[dashed, blue, bend right=80, very thick] (6, -2) to (8.5, -2);
\draw[dashed, blue, bend right=80, very thick] (5.5, -2) to (9, -2);
\draw[dashed, red, very thick, bend right=30] (4.25,-2) to (5.25, -3);
\draw[dashed, green!60!black, very thick, bend right=30] (7.25,2) to (6.25, 3);

\draw[very thick, green!60!black] (1,1) to (1,-1);
\draw[very thick, green!60!black, bend left=17] (1.5,-1) to (7,-1);
\draw[very thick, green!60!black] (7.5,-1) to (7.5, -0.8) to (1.5,0.2) to (1.5,1); 
\draw[very thick, green!60!black] (7,1) to (7.5,1); 

\draw[very thick, blue, bend left=18] (2.5, -1) to (6,-1);
\draw[very thick, blue, bend left=15] (3, -1) to (5.5,-1);
\draw[very thick, blue, bend right=15] (6,1) to (8.5,1);
\draw[very thick, blue, bend right=18] (5.5, 1) to (9,1);
\draw[very thick, blue] (9, -1) to (9,-0.5) to (3,0.5) to (3,1);
\draw[very thick, blue] (8.5, -1) to (8.5,-0.7) to (2.5,0.3) to (2.5,1);

\draw[very thick, red] (10,-1) to (10,-0.4) to (4,0.7) to (4,1);
\draw[very thick, red, bend right=17] (4.5,1) to (10,1);
\draw[very thick, red] (10.5,1) to (10.5,-1);
\draw[very thick, red] (4,-1) to (4.5, -1);

\foreach \x in {1,...,7}{
\node at (1.5*\x-0.5,1) [circle, draw, fill, inner sep=1pt] {};
\node at (1.5*\x-0.5,-1) [circle, draw, fill, inner sep=1pt] {};
\node at (1.5*\x,1) [circle, draw, fill, inner sep=1pt] {};
\node at (1.5*\x,-1) [circle, draw, fill, inner sep=1pt] {};
}
\end{tikzpicture}
\hspace{3mm}
\begin{tikzpicture}[scale=0.55]
\draw[thick] (0,0) to (11.5,0);
\draw[thick] (0,0.1) to (0,-0.1);
\draw[thick] (11.5,0.1) to (11.5,-0.1);

\draw[very thick, green!60!black] (1,2) to (1,1);
\draw[very thick, green!60!black] (1,-2) to (1,-1);
\draw[very thick, green!60!black] (1.5,-1) to (1.5,-2);
\draw[very thick, green!60!black] (1.5,1) to (1.5,2);

\draw[very thick, blue] (2.75,1) to (2.75,2);
\draw[very thick, blue] (2.75,-1) to (2.75,-2);

\draw[very thick, red] (4,-2) to (4,-1);
\draw[very thick, red] (4.5,-2) to (4.5,-1);
\draw[very thick, red] (4,1) to (4,2);
\draw[very thick, red] (4.5,1) to (4.5,2);

\draw[very thick, blue] (5.75,1) to (5.75,2);
\draw[very thick, blue] (5.75,-2) to (5.75,-1);

\draw[very thick, green!60!black] (7,1) to (7,2);
\draw[very thick, green!60!black] (7.5,1) to (7.5,2);
\draw[very thick, green!60!black] (7,-2) to (7,-1);
\draw[very thick, green!60!black] (7.5,-2) to (7.5,-1);

\draw[very thick, blue] (8.75,1) to (8.75,2);
\draw[very thick, blue] (8.75,-2) to (8.75,-1);

\draw[very thick, red] (10,1) to (10,2);
\draw[very thick, red] (10.5,1) to (10.5,2);
\draw[very thick, red] (10,-2) to (10,-1);
\draw[very thick, red] (10.5,-2) to (10.5,-1);

\draw[dashed, blue, bend left=80, very thick] (2.75, 2) to (5.75, 2);
\phantom{\draw[dashed, blue, bend left=80, very thick] (2.5, 2) to (6, 2);}
\draw[dashed, blue, bend right=80, very thick] (5.75, -2) to (8.75, -2);
\phantom{\draw[dashed, blue, bend right=80, very thick] (5.5, -2) to (9, -2);}

\draw[dashed, red, very thick, bend right=30] (4.25,-2) to (5.25, -3);
\draw[dashed, green!60!black, very thick, bend right=30] (7.25,2) to (6.25, 3);
\draw[very thick, green!60!black] (1,1) to (1,-1);
\draw[very thick, green!60!black, bend left=17] (1.5,-1) to (7,-1);
\draw[very thick, green!60!black] (7.5,-1) to (7.5, -0.8) to (1.5,0.2) to (1.5,1); 
\draw[very thick, green!60!black] (7,1) to (7.5,1); 

\draw[very thick, blue, bend left=17] (2.75, -1) to (5.75,-1);
\draw[very thick, blue, bend right=15] (5.75,1) to (8.75,1);
\draw[very thick, blue] (8.75, -1) to (8.75,-0.7) to (2.75,0.3) to (2.75,1);

\draw[very thick, red] (10,-1) to (10,-0.4) to (4,0.7) to (4,1);
\draw[very thick, red, bend right=17] (4.5,1) to (10,1);
\draw[very thick, red] (10.5,1) to (10.5,-1);
\draw[very thick, red] (4,-1) to (4.5, -1);

\foreach \x in {1,3,5,7}{
\node at (1.5*\x-0.5,1) [circle, draw, fill, inner sep=1pt] {};
\node at (1.5*\x-0.5,-1) [circle, draw, fill, inner sep=1pt] {};
\node at (1.5*\x,1) [circle, draw, fill, inner sep=1pt] {};
\node at (1.5*\x,-1) [circle, draw, fill, inner sep=1pt] {};
}
\foreach \x in {2,4,6}{
\node at (1.5*\x-0.25,1) [circle, draw, fill, inner sep=1pt] {};
\node at (1.5*\x-0.25,-1) [circle, draw, fill, inner sep=1pt] {};
}
\end{tikzpicture}
\end{center}
\vspace{-5mm}
\caption{Patching scheme for three \noncrossing tours in the two considered crossing patterns.}\label{fig:tricolorpatching}
\end{figure}
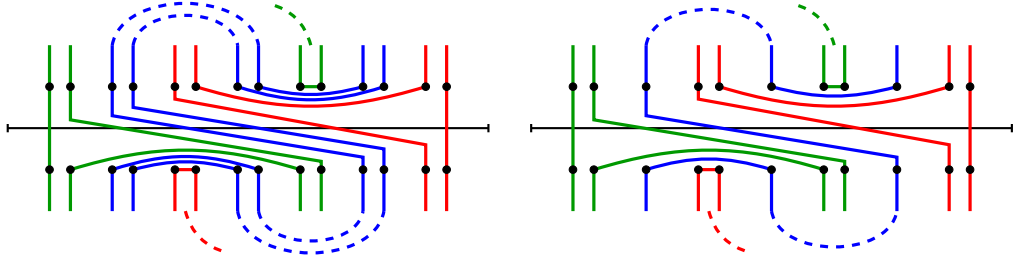

\includeinappendix{
\begin{proof}[Proof sketch]
We modify the three curves in several steps to reduce the number of crossings with~$s$.
We arrange the crossings into \emph{groups}, where a group is a maximal set of consecutive monochromatic crossings.
The road map for our proof is roughly as follows: First, we reduce the number of crossings inside a group by applying Lemma~\ref{lem:patching-one-tour}. 
Next, we further simplify the occurring patterns by applying~Lemma~\ref{lem:patching-two-tours}, eliminating longer sequences where only two of the three colors appear.
Then, we give a new patching scheme for the tricolored subpatterns illustrated in Figure~\ref{fig:possible-connections}. For this, the first step is to investigate how the crossings are connected via the tours outside of $N_{\delta}(s)$ and we show that the connections illustrated by the dashed lines in Figure~\ref{fig:possible-connections}  must exist (up to symmetry). 
Then, one can observe that modifying the tours as illustrated in Figure~\ref{fig:tricolorpatching} gives three closed curves.
Last, we show that after repeatedly applying these patchings schemes whenever one of the described subpatterns appears, the remaining number of crossings is at most 18 and the cost of redirecting the tours is at most $79\cdot \l(s)$.
\end{proof}
}

\subsection{Structure theorem for three \noncrossing tours}
Now, we have all the prerequisites in place to state and prove our structure theorem (cf.~Theorem~\ref{thm:intro-structure-theorem}) for three \noncrossing tours. 

\newcommand{\crossingnumber}{m}
For this, given an instance of $\ktsp$ with terminals in $\L$ and a shift vector $\vec{a}\in\Lmin$, we say that a solution $\Pi=(\pi_c)_{c\in C}$ is \emph{$(r,\crossingnumber, \delta)$-portal respecting} if, for every boundary $b$ in $D(\vec{a})$, the intersection points $b \cap \bigcup_{c\in C} \pi_c$ are contained in the the portals $\grid(b, r \log L, \delta)$ and every portal is intersected in at most $\crossingnumber$  points in total.%

\begin{restatable}[Structure Theorem for $\ttsp$]{theorem}{structuretheorem}\label{thm:structure-thm}
Let an instance of $\ttsp$ with $T\subseteq \L$ and $\epsilon>0$ be given. Then there exists a shift vector $\vec{a}\in \Lmin$ and $\delta>0$ such that there is either a $(\lceil (15\sqrt{2}+4)/\epsilon \rceil, 18, \delta)$-portal respecting solution of cost at most $(1+\epsilon)\cdot \opt$, or there is a $(\lceil (15\sqrt{2}+4)/\epsilon \rceil, 18, \delta)$-portal respecting \twotoursol $(\pi_1,\pi_2)$ with $2\l(\pi_1)+\l(\pi_2)\leq \left(\frac{5}{3}+\frac{\epsilon}{2} \right)\cdot \opt$.
\end{restatable}

\newcommand{\epsone}{\eta}

\def\proofstructuretheorem{
\begin{proof}
Let $\epsone>0$ be a small real number that will
be appropriately set later on in the proof. Let $\Pi=(\pir, \pig, \pib)$ be a solution of cost at most $\left(1+\epsone\right)\cdot \opt$.
Recall that we can assume wlog.~that each tour in $\Pi$ consists of straight line segments connecting points in $N_{\frac{1}{4}}(\L)$.
Choose a shift vector $\vec{a}\in\Lmin$ uniformly at random. Let $\delta>0$ be a small real number that will be appropriately set later on in the proof. Set $r:=\lceil (15\sqrt{2}+4)/\epsilon \rceil$, and consider the dissection $D(\vec{a})$ with the portals $\grid(b,r \log L,\delta)$ placed on every boundary $b$.
Throughout this proof a \emph{crossing} is a point in the intersection of the (partially modified) tours and the grid lines.

As the first step, we move all crossings to portals by applying the modifications of Lemma~\ref{lem:portals} to $\pir, \pig, \pib$ where we set the value $\delta'$  to be less than 0.25
and let $\hat{\Pi}=(\hat{\pi}_{\mathrm{R}}, \hat{\pi}_{\mathrm{G}}, \hat{\pi}_{\mathrm{B}})$ be the resulting tours. The tours remain disjoint and, 
since they were only modified in a $\delta'$-neighbourhood around the grid lines and terminals have distance at least 0.25 from any grid line, they still visit all the terminals so that $\hat{\pi}_{\mathrm{R}}, \hat{\pi}_{\mathrm{G}}, \hat{\pi}_{\mathrm{B}}$ is still a solution to the given instance of $\ttsp$. By Lemma~\ref{lem:portals}, we have
\begin{equation}\label{eq:length-pi-hat}
\mathbb{E} \left[ \cost(\hat{\Pi}) \right] \leq \left( 1+ \frac{7\sqrt{2}}{r} \right) \cdot \left(\l(\pir)+\l(\pig)+\l(\pib) \right)
\leq \left( 1 + \epsone+ \frac{7\sqrt{2}\cdot (1+\epsone)}{r} \right) \cdot \opt.
\end{equation}
Moreover, no crossing of $\hat{\pi}_{\mathrm{R}}, \hat{\pi}_{\mathrm{G}}, \hat{\pi}_{\mathrm{B}}$ with grid lines lies in the intersection of two grid lines. Note that this implies the following: 
Given a boundary $b$, consider the first portal $s$ on $b$. Recall that it is placed such that its left endpoint $\vec{x}$ lies on the left endpoint of $b$ and $\vec{x}$ lies in the intersection of $b$ with a grid line $g$ perpendicular to $b$. Since no crossing lies in the intersection of two grid lines, there is no crossing that equals $\vec{x}$. 
Therefore, if we are about to apply patching (i.e., Lemma~\ref{lem:patching}) on $s$, we can choose %
a subsegment $s'$ of $s$ not containing $\vec{x}$ and apply patching on $s'$ so that the the tours remain unchanged in some neighbourhood of $g$.
The last portal on $b$ can be patched analogously.
This means that we can patch portals in a way that grid lines of lower levels are not affected.

Wlog., assume $\hat{\pi}_{\mathrm{R}}, \hat{\pi}_{\mathrm{G}}\leq \hat{\pi}_{\mathrm{B}}$. We further modify the tours in $\hat{\Pi}$ by applying the following procedure:
consider every boundary $b$ one by one in non-decreasing order of their levels and consider every portal $s$ on $b$.
If there is no red crossing next to a green crossing on $b$, we apply patching (i.e., Lemma~\ref{lem:patching}) on $s$ (as described in the previous paragraph).
If there is a red crossing next to a green crossing, we connect the red and green tour along $s$ and, for the remainder of the procedure, we identify the colors red and green with each other (i.e., we are left with only two tours). Note that this enables us to apply patching. 

We argue that the resulting tours are $(r, 18, \delta)$-portal respecting.
Similarly as in the proof of Lemma~\ref{lem:portals}, note that patching can create additional crossings on a grid line perpendicular to the portal under consideration. However, we have argued above that only grid lines crossing $s^\circ$ are affected and, by Observation~\ref{obs:level-intersections}, these grid lines are of higher levels so that the corresponding boundaries have not been considered yet. Moreover, the additionally created crossings on a boundary $b'$ perpendicular to $b^\circ$ lie inside a portal because we only change the tours in a sufficiently small neighbourhood (smaller than the portal length) of $b$ and $b\cap b'$ is an endpoint of $b'$ so  it lies in a portal on $b'$.
With this, we obtain that the obtained tours are $(r,18, \delta)$-portal respecting for every choice of $\Vec{a}$.

Next, we bound the cost of the resulting tours.
It follows from Lemma~\ref{lem:patching} that the cost of patching in a single portal is $O(\delta)$. %
Since there are $O( L^2  r \log L )$ portals, the total cost of patching is $O( \delta L^2  r \log L )$, so $\delta$ can be chosen small enough to make the total cost of patching at most $\opt /r$.
To summarize, if there was never a red crossing next to a green crossing in a portal, we obtain a $(r,18,\delta)$-portal respecting solution of cost at most
\begin{equation*}
\mathbb{E}_{\vec{a}} \left[\cost(\hat{\Pi})\right]+\frac{\opt}{r}
\overset{\eqref{eq:length-pi-hat}}{\leq} 
\left( 1 + \epsone+ \frac{7\sqrt{2}\cdot (1+\epsone)+1}{r} \right) \cdot \opt
\leq (1+\epsilon) \cdot \opt,
\end{equation*}
where we have used in the last inequality that $r\geq (15\sqrt{2}+4)/\epsilon$ and $\epsone$ can be chosen small enough.
In particular, at any time during the procedure in which the red and green tour are not connected, the total length of any subset of the three tours is increased at most by $\opt/r$.

Now, assume that, at some point in the procedure where a portal $s$ is considered, there is a red crossing next to a green crossing.
The cost of connecting the red and green tour is then at most $2\delta$, which is at most $\opt/r$ for $\delta$ small enough.
Therefore, the total cost of patching and connecting the tours is at most $2\opt/r$.
Therefore, the result of the procedure is a \twotoursol $(\pi_1,\pi_2)$ that is $(r,18,\delta)$-portal respecting with
\begin{align*}
2\l(\pi_1)+\l(\pi_2) &
\leq 2 \left( \l(\hat{\pi}_{\mathrm{R}})+ \l(\hat{\pi}_{\mathrm{G}}) \right) +  \l(\hat{\pi}_{\mathrm{B}}) + 2 \cdot \frac{2\opt}{r}\\
\overset{\l(\hat{\pi}_{\mathrm{R}}),\l(\hat{\pi}_{\mathrm{G}})\leq  \l(\hat{\pi}_{\mathrm{B}})}&{\leq} \frac{5}{3} \cdot \left( \l(\hat{\pi}_{\mathrm{R}})+ \l(\hat{\pi}_{\mathrm{G}}) +  \l(\hat{\pi}_{\mathrm{B}}) \right)+ \frac{4\opt}{r}\\
\overset{\eqref{eq:length-pi-hat}}&{\leq}  \frac{5}{3} \cdot  \left( 1 + \epsone+ \frac{7\sqrt{2}\cdot (1+\epsone)+4}{r} \right) \cdot \opt
\leq \left( \frac{5}{3}+\frac{\epsilon}{2} \right) \cdot \opt,
\end{align*}
where we have used in the last inequality that $r\geq (15\sqrt{2}+4)/\epsilon$ and that $\epsone$ is chosen small enough.

Last, recall that $\vec{a}$ was chosen uniformly at random from $\Lmin$. 
We have seen that we obtain a portal-respecting solution for every choice of $\vec{a}$ and the resulting cost is in expectation as desired.
Using the probabilistic method, this implies that there is a vector $\vec{a}\in\Lmin$ such that 
the obtained tours have the cost stated in the Theorem.
\end{proof}
}

\movetoappendix{\proofstructuretheorem}
\includeinappendix{
\begin{proof}[Proof sketch]
First, we move all crossings to portals by applying Lemma~\ref{lem:portals}. Then, we consider boundaries $b$ one by one in non-decreasing order of their levels and apply patching (Lemma~\ref{lem:patching}). Note that patching can create new crossings on another boundary perpendicular to $b$. We use Observation~\ref{obs:level-intersections} to argue that this other boundary has not been modified yet. Since we can choose the portal-length $\delta$ arbitrarily small, the total cost of patching is negligible.
\end{proof}
}

Recall that, given a \twotoursol, one can ``double'' one of the tours to obtain an induced two-tour solution for $\ttsp$ (cf.~Observation~\ref{obs:two-color-solutions}). Applying this to a $(\lceil (15\sqrt{2}+4)/\epsilon \rceil, 18, \delta)$-portal respecting \twotoursol obtained from Theorem~\ref{thm:structure-thm}, this gives a $(\lceil (15\sqrt{2}+4)/\epsilon \rceil, 36, \delta)$-portal respecting solution for $\ttsp$.
Combining this with Theorem~\ref{thm:structure-thm}, we obtain the following.

\begin{corollary}\label{cor:structure-theorem}
For every instance of $\ttsp$ with terminals in $\L$ and $\epsilon>0$, there is a shift vector $\vec{a}\in \Lmin$ and $\delta>0$ such that there exists a $(\lceil (15\sqrt{2}+4)/\epsilon \rceil, 36, \delta)$-portal respecting solution of cost at most $\left(\frac{5}{3}+\epsilon \right)\cdot \opt$.
\end{corollary}

\section{A dynamic programming algorithm}\label{sec:dynamic-programming}

In the previous section, we have seen that there is a $\left(\frac{5}{3}+\epsilon\right)$-approximate portal-respecting solution.
In this section, we give a polynomial-time algorithm that computes an ``optimal'' (in the sense of Theorem~\ref{thm:intro-dynamic-program}) portal-respecting solution.

Our algorithm is based on the same ideas as Arora's dynamic programming algorithm for Euclidean TSP \cite{arora} and the algorithm by Dross et al.~for $\twotsp$ \cite{dross}.
The difference to our work is that we need to solve a more general problem: First, we allow for any fixed number of colors of terminals and search for \noncrossing tours. Second, we have weighted colors, i.e., the tours of different colors contribute differently to the total cost.

More precisely, by \emph{$\ktsp'$} we denote the following problem: 
a set $C$ of $k$ colors is given together with an integer $L$ that is a power of two. The input consists of a set of terminals $T_c \subseteq \L$ for each color $c \in C$, a color weight $w_c\geq 0$ for each $c \in C$, a shift vector $\vec{a} \in \Lmin$, $\delta>0$ (sufficiently small) and two integers $r,\crossingnumber \in \N$. 
We consider the dissection $D(\vec{a})$ and, as before, we place $r \log L$ portals on every boundary. 
A \emph{solution} to $\ktsp'$ is a $k$-tuple of tours $\Pi=(\pi_c)_{c \in C}$ such that
every terminal is visited by the tour of the same color (i.e., $T_c \subseteq \pi_c$ for every $c \in C$),
the tours are pairwise disjoint, and  $(r,\crossingnumber, \delta)$-portal respecting. 
The \emph{cost of a solution} for $\ktsp'$ is then $\cost(\Pi):=\sum_{c \in C} w_c \cdot \l(\pi_c)$. 
Similarly as for $\ktsp$ (cf.~Figure~\ref{fig:delta-optimal}), a solution minimizing the cost does not necessarily exist.
By $\opt:=\inf \{\cost(\Pi)): \Pi \text{ is a solution}\}$, we denote the value that we want to approximate.

The aim of this section is to prove the following theorem.

\begin{restatable}{theorem}{thmdynamicprogramming}\label{thm:dynamic-programming}
There is an algorithm that computes a parametric solution $\Pi(\lambda)$ to $\ktsp'$ in time $L^{O(\crossingnumber r \log k)}$ such that $\lim_{\lambda \to 0} \cost\left(\Pi(\lambda)\right)= \opt$.
\end{restatable}

The main idea is to use dynamic programming. More precisely, we consider the nodes of $D(\vec{a})$ (i.e., squares in $\R^2$) one by one from leaves to the root and consider all possible ways that a solution can cross the border of the square. For a non-leaf node, we will find these possibilities by combining the solutions for the four subsquares.

\subsection{The multipath problem and the lookup table}

In this subsection, we investigate how the problem of finding a portal-respecting solution can be decomposed into subproblems. For this, fix an instance $(T_c)_{c\in C}, (w_c)_{c\in C},\vec{a},\delta,r,m$ of $\ktsp'$.

Recall that portals are subsegments of boundaries of $D(\vec{a})$ of positive length $\delta$. %
For every portal in $D(\vec{a})$, we place $\crossingnumber$ distinct points in it, %
called \emph{subportals}, and we will consider solutions that only intersect boundaries in subportals. We place the subportals in a way that no subportal lies at the intersection of the gridlines. %
Note that each subportal is crossed at most once because the tours are pairwise disjoint and a subportal is a point.
Observe that every $(r,\crossingnumber, \delta)$-portal-respecting solution can be transformed into a subportal-respecting solution at cost at most $2\delta$. Since $\delta$ will be chosen arbitrarily small, this cost is negligible. Therefore, the requirement that a solution only intersects portals in the subportals is not a restriction.

Let a node (i.e., a square) $S$ of $D(\vec{a})$ be given.
We consider all possible ways that the tours of a solution can leave and enter $S$ through the subportals. For this, we color each subportal with one of the colors in $C$ or leave it uncolored. This encodes which of the tours in a solution crosses through the subportal where uncolored means that none of them crosses.
We denote a coloring on all the subportals associated with $\partial S$ by a $k$-tuple $(P_c)_{c\in C}$ of pairwise disjoint subsets of the subportals.
Given such a coloring, we also have to specify which subportals are crossed consecutively by a tour. For this, we consider matchings between the colored subportals (cf.~Figure~\ref{fig:multipath}). The fact that tours of different colors are \noncrossing and that we can restrict ourselves to constructing simple tours\footnote{Assume that one of the tours $\pi_c$ is not simple and let $\vec{x}$ be a point in which $\pi_c$ crosses itself. Then, $\pi_c$ can be modified in a sufficiently small neighbourhood of $\vec{x}$ such that $\pi_c$ is simple in that neighbourhood, still does not intersect any of the other tours, and its length is decreased.}, significantly reduces the number of matchings that need to be considered. 

For this, we define a \emph{\noncrossing matching} of a coloring $(P_c)_{c\in C}$ as a $k$-tuple $(M_c)_{c \in C}$ such that the following holds: For every $c \in C$, $M_c$ is a partition of $P_c$ into sets of size two, and there are pairwise disjoint curves $(\pi_{\vec{p}\vec{q}})$ for every ${\{\vec{p},\vec{q}\} \in \bigcup_{c \in C} M_c}$ such that each~$(\pi_{\vec{p}\vec{q}})$ connects the subportals $\vec{p}$ and $\vec{q}$ (cf.~Figure~\ref{fig:multipath}). Observe that the latter condition is fulfilled if and only if the the set of straight line segments $\{ \overline{\vec{p} \vec{q}}: \{\vec{p}\vec{q}\} \in \bigcup_{c \in C} M_c\}$ is pairwise disjoint.
We will see later that the number of \noncrossing matchings is bounded by $L^{O(\crossingnumber r)}$ (cf.~Lemma~\ref{lem:noncrossing-matchings}).

If $S$ contains all terminals of a color $c$ or none of the terminals of color $c$, the tour $\pi_c$ of a solution does not necessarily intersect the border of $S$ (but note that it might intersect it to bypass a terminal of another color). For this reason, we also allow $P_c=\emptyset$. In this case, we need the constructed curves to form a single cycle. However, if $S$ contains one but not all terminals of color $c$, a solution must intersect the border of $S$ so  we only allow for $P_c \neq \emptyset$. We also need to ensure that the tours $\pi_c$ that we construct in the end are connected (i.e., do not consist of several disconnected cycles). For this reason, in the case $P_c \neq \emptyset$, we will require $\pi_c \cap S$ to consist of curves starting and ending at the border of $S$.

Now, we have all the prerequisites in place to define the \emph{multipath problem} (cf.~Figure~\ref{fig:multipath}): Fix an instance of $\ktsp'$ and a node $S$ of $D(\vec{a})$. 
The input consists of a coloring $(P_c)_{c\in C}$ of the subportals of $\partial S$ and a \noncrossing matching $(M_c)_{c \in C}$ of $(P_c)_{c\in C}$, where $P_c=\emptyset$ is only a valid input if $S \cap T_c \in \{T_c,\emptyset\}$.
A solution to the multipath problem for $S$ is a $k$-tuple $(\Pi_c)_{c \in C}$ of sets of simple curves in $S$ such that:
\begin{enumerate}[a)]
\item every terminal of color $c$ in $S$ is visited by a curve of color $c$, i.e., $T_c \cap S \subset \bigcup_{\pi \in \Pi_c} \pi$,
\item the curves are disjoint,%
\item for every $c \in C$ and every $\{\vec{p}, \vec{q}\} \in M_c$, there is a curve $\pi \in \Pi_c$ connecting the two subportals and not crossing the border of $S$ otherwise, i.e.,  $\{\vec{p}, \vec{q}\} = \pi \cap \partial S$,
\item the border of $S$ is only crossed at colored subportals, i.e., $\partial S \cap \left( \bigcup_{c \in C, \pi \in \Pi_c} \pi \right) \subseteq \bigcup_{c \in C} P_c$,
\item for every $c \in C$, if $P_c\neq \emptyset$, each curve of color $c$ connects two subportals of color $c$ on $\partial S$,
\item for every $c \in C$, if $P_c = \emptyset$, then $\bigcup_{\pi\in\Pi_c} \pi$ is a single closed curve.
\end{enumerate}
Observe that conditions c) and d) together imply that every colored subportal is crossed exactly once by a curve of the same color and there are no other crossings on the border of $S$.
The \emph{cost} of a solution to the multipath problem is $\sum_{c \in C} w_c \cdot \l(\Pi_c)$, where $ \l(\Pi_c):=\sum_{\pi \in \Pi_c} \l(\pi)$ denotes the total length of the set of curves.
Note that, similarly as for $\ktsp$ and $\ktsp'$, a solution minimizing the cost does not necessarily exist.

\begin{figure}
\begin{center}
\begin{tikzpicture}[scale=0.62]
\draw[-, very thick] (0,0) to (4,0) to (4,4) to (0,4) to (0,0);
\node at (2,2) [circle, draw, fill, red, inner sep=0pt, minimum size=6pt] {};
\node at (0,1.15) [rectangle, draw, fill, red, inner sep=0pt, minimum size=4pt] {};
\node at (0,0.85) [rectangle, draw, fill, blue, inner sep=0pt, minimum size=4pt] {};
\node at (0,2.65) [rectangle, draw, fill, red, inner sep=0pt, minimum size=4pt] {};
\node at (0.25,0) [rectangle, draw, fill, blue, inner sep=0pt, minimum size=4pt] {};
\node at (0.55,0) [rectangle, draw, fill, green!60!black, inner sep=0pt, minimum size=4pt] {};
\node at (2.15,0) [rectangle, draw, fill, green!60!black, inner sep=0pt, minimum size=4pt] {};
\node at (4,1.15) [rectangle, draw, fill, green!60!black, inner sep=0pt, minimum size=4pt] {};
\node at (0.25,4) [rectangle, draw, fill, red, inner sep=0pt, minimum size=4pt] {};
\node at (0.85,4) [rectangle, draw, fill, red, inner sep=0pt, minimum size=4pt] {};
\node at (1.15,4) [rectangle, draw, fill, green!60!black, inner sep=0pt, minimum size=4pt] {};
\draw[red, thick] (0,2.65) to (0.25,4);
\draw[red, thick] (0,1.15) to (2,2) to (0.85,4);
\draw[blue, thick] (0, 0.85) to (0.25,0);
\draw[green!60!black, thick] (0.55,0) to (2.35,1.92) to (1.15,4);
\draw[green!60!black, thick] (2.15,0) to (4,1.15);
\end{tikzpicture}
\hspace{10mm}
\begin{tikzpicture}[scale=0.62]
\draw[-, very thick] (-4,0) to (4,0) to (4,8) to (-4,8) to (-4,0);
\draw[-, very thick] (0,0) to (0,8);
\draw[-, very thick] (-4,4) to (4,4);
\node at (2,2) [circle, draw, fill, red, inner sep=0pt, minimum size=6pt] {};
\node at (2,2) [circle, draw, fill, red, inner sep=0pt, minimum size=6pt] {};
\node at (0,1.15) [rectangle, draw, fill, red, inner sep=0pt, minimum size=4pt] {};
\node at (0,0.85) [rectangle, draw, fill, blue, inner sep=0pt, minimum size=4pt] {};
\node at (0,2.65) [rectangle, draw, fill, red, inner sep=0pt, minimum size=4pt] {};
\node at (0.25,0) [rectangle, draw, fill, blue, inner sep=0pt, minimum size=4pt] {};
\node at (0.55,0) [rectangle, draw, fill, green!60!black, inner sep=0pt, minimum size=4pt] {};
\node at (2.15,0) [rectangle, draw, fill, green!60!black, inner sep=0pt, minimum size=4pt] {};
\node at (4,1.15) [rectangle, draw, fill, green!60!black, inner sep=0pt, minimum size=4pt] {};
\node at (0.25,4) [rectangle, draw, fill, red, inner sep=0pt, minimum size=4pt] {};
\node at (0.85,4) [rectangle, draw, fill, red, inner sep=0pt, minimum size=4pt] {};
\node at (1.15,4) [rectangle, draw, fill, green!60!black, inner sep=0pt, minimum size=4pt] {};
\draw[red, thick] (0,2.65) to (0.25,4);
\draw[red, thick] (0,1.15) to (2,2) to (0.85,4);
\draw[blue, thick] (0, 0.85) to (0.25,0);
\draw[green!60!black, thick] (0.55,0) to (2.35,1.92) to (1.15,4);
\draw[green!60!black, thick] (2.15,0) to (4,1.15);
\node at (2,6) [circle, draw, fill, red, inner sep=0pt, minimum size=6pt] {};
\draw[-, red, thick] (0.25,4) to (2,6) to (0.85,4);
\node at (-2,6) [circle, draw, fill, green!60!black, inner sep=0pt, minimum size=6pt] {};
\draw[thick, green!60!black] (-4,6) to (0,6);
\draw[thick, green!60!black] (0,6) to (1,8);
\draw[thick, green!60!black] (1.15,4) to (4,5.5);
\draw[thick, blue] (0,0.85) to (-3,0);
\draw[thick, red] (0, 1.15) to (-4,1);
\draw[thick, red] (0, 2.65) to (-4,5);
\node at (-4, 1) [rectangle, draw, fill, red, inner sep=0pt, minimum size=4pt] {};
\node at (-4,5) [rectangle, draw, fill, red, inner sep=0pt, minimum size=4pt] {};
\node at (-2.298,4) [rectangle, draw, fill, red, inner sep=0pt, minimum size=4pt] {};
\node at (-3,0) [rectangle, draw, fill, blue, inner sep=0pt, minimum size=4pt] {};
\node at (-4,6) [rectangle, draw, fill, green!60!black, inner sep=0pt, minimum size=4pt] {};
\node at (0,6) [rectangle, draw, fill, green!60!black, inner sep=0pt, minimum size=4pt] {};
\node at (1,8) [rectangle, draw, fill, green!60!black, inner sep=0pt, minimum size=4pt] {};
\node at (4,5.5) [rectangle, draw, fill, green!60!black, inner sep=0pt, minimum size=4pt] {};
\end{tikzpicture}
\caption{Illustration of the multipath problem: The cycles denote terminals and the rectangles denote subportals. On the left, a solution of a base case is given. On the right, we are given a combination of compatible entries and the resulting solution.}\label{fig:multipath}
\end{center}
\end{figure}
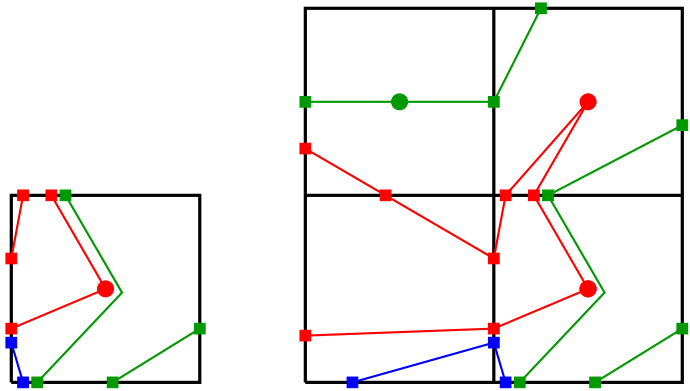

\newcommand{\lt}{\mathrm{LT}_S}

We define the \emph{lookup-table} $\lt$ of $S$ as follows:
\begin{align*}
\lt \left[(P_c)_{c\in C},(M_c)_{c \in C}\right]&:=\inf \Bigl\{\sum_{c \in C} w_c \cdot \l(\Pi_c): (\Pi_c)_{c\in C} \text{ is a solution to the multipath}\\
& \hspace{1.3cm} \text{problem with inputs $(P_c)_{c\in C}$ and $(M_c)_{c \in C}$}\Bigl\},
\end{align*}
where $(P_c)_{c\in C}$ is a coloring of the subportals of $S$ and $(M_c)_{c \in C}$ is a \noncrossing matching for $(P_c)_{c\in C}$, where $P_c=\emptyset$ is only allowed if $T_c \cap S \in \{\emptyset,S\}$.

As we will calculate every entry of the lookup-tables in a dynamic programming fashion, it is important to bound the
number of entries of $\lt$.
First, we bound the number of \noncrossing matchings.\includeinappendix{The key idea for this is to upper bound them by the Catalan numbers.}

\begin{restatable}{lemma}{lemnoncrossingmatchings}\label{lem:noncrossing-matchings}
For a node $S$ of $D(\vec{a})$ and a coloring $(P_c)_{c\in C}$ of the subportals in $\partial S$, the number of \noncrossing matchings is at most 
$L^{O(\crossingnumber r)}$ and we can list all the \noncrossing matchings in time $L^{O(\crossingnumber r)}$.%
\end{restatable}

\def\prooflemnoncrossingmatchings{
\begin{proof}
Let a coloring $(P_c)_{c\in C}$ of the subportals of $S$ be given.
We count the number of \noncrossing matchings of $\bigcup_{c \in C} P_c$, i.e., we ignore the colors and allow for all colored subportals to be matched as long as the straight line segments 
$\{ \overline{p_1p_2}: \{\vec{p}_1,\vec{p}_2\} \text{ is in the matching}\}$ are \noncrossing.
It is immediate that this is an upper bound on the number of \noncrossing matchings of $(P_c)_{c\in C}$. 

Let $i:=|\bigcup_{c \in C} P_c|$ denote the number of colored subportals and let $f(i)$ denote the number of \noncrossing matchings between them. 
Assume a \noncrossing matching contains the pair of subportals $\{\vec{p},\vec{p}'\}$.
Note that $\overline{\vec{p}\vec{p}'}$ separates the square into two faces.
Therefore, the matching cannot contain a pair $\{\vec{q},\vec{q}'\}$ where $\vec{q}$ and $\vec{q}'$ lie in different faces.
Hence, if there are $j$ subportals different from $\vec{p}, \vec{p}'$ in one of the faces, there are $i-j-2$ subportals different from $\vec{p}, \vec{p}'$ in the other face, and the number of \noncrossing matchings that contain the pair $\{\vec{p},\vec{p}'\}$ is $f(i)\cdot f(i-j-2)$.
Summing up over all choices for $\vec{p}'$, we obtain
\begin{equation}
f(i) \leq \sum\limits_{j=0}^{i-1} f(j) \cdot f(i-j).
\end{equation}
Note that this sum is precisely the recursive definition of the well-known \emph{Catalan numbers} (see, e.g.,~\cite{catalan}). Additionally using that $f(0)=1$ and $f(1)=0$, we obtain that $f(i)$ is upper bounded by the $i$-th Catalan number $C_i$. Since $C_i=O(4^i)$ (cf.~\cite[Theorem 3.1]{catalan}), we have
\begin{equation*}
f(i)\leq C_i = O(4^i).
\end{equation*}
As $i$ is the number of colored portals, we have $i=O(\crossingnumber r \log L)$. 
Together, this gives that the number of \noncrossing matchings is upper bounded by 
\begin{equation*}
O(4^{O(\crossingnumber r \log L)})=2^{O(\crossingnumber r \log L)}=L^{O(\crossingnumber r)}.
\end{equation*}
Note that the above recursion also gives an efficient recursive procedure for listing all the non-crossing matchings.%
\end{proof}
}\movetoappendix{\prooflemnoncrossingmatchings}

Next, note that the number of subportals on the border of~$S$ is $O(\crossingnumber r \log L)$ so that the number of possible colorings of the subportals of~$S$ is $(k+1)^{O(\crossingnumber r \log L)}=L^{O(\crossingnumber r \log k)}$, and we can easily list all colorings in $L^{O(\crossingnumber r \log k)}$ time as well.
Combining this with Lemma~\ref{lem:noncrossing-matchings}, we obtain the following result.

\begin{observation}\label{obs:lookup-table-size}
For every node $S$ of $D(\vec{a})$, the lookup-table $\lt$ has
$\leq L^{O(\crossingnumber r \log k)}$
entries.
\end{observation}

\subsection{The algorithm}

Next, we show that the entries of the lookup-table 
can be computed in polynomial time. 
We begin with computing the tables of leaves.

\begin{restatable}{lemma}{lemlookuptableleaf}\label{lem:lookup-table-leaf}
For every leaf $S$ of $D(\vec{a})$, the lookup-table $\lt$ can be computed in time
$L^{O(\crossingnumber r \log k)}$.
\end{restatable}

\def\prooflemlookuptableleaf{
\begin{proof}
Consider a leaf of $D(\vec{a})$, i.e., a square $S$ of side length 1, and an instance of the multipath problem, i.e., a coloring of the subportals $(P_c)_{c\in C}$ and a \noncrossing matching $(M_c)_{c \in C}$.
Recall that terminals are distinct points in $\L$ so that $S$ contains at most one terminal and, if $S$ contains a terminal, it lies in the center of $S$.

If it contains no terminal, observe that an optimal solution is simply given by $\Pi_c=\{ \overline{\vec{p}_1\vec{p}_2}: \{\vec{p}_1, \vec{p}_2\} \in M_c\}$ ($c \in C$).
Therefore, assume from now on that $D(\vec{a})$ does contain a terminal~$\vec{t}$, say of color $c^* \in C$. 
If there is no subportal colored $c^*$ (note that this is only possible if $\vec{t}$ is the only portal of color $c$), consider the solution given by $\Pi_c=\{ \overline{\vec{p}_1\vec{p}_2}: \{\vec{p}_1, \vec{p}_2\} \in M_c\}$ ($c \neq c^*$) and additionally letting $\Pi_{c^*}$ consist of a single small cycle of length at most $\lambda$ that visits $\vec{t}$ and does not intersect any of the tours in $\Pi_c, c \neq c^*$. In case one of the straight line segments in $\Pi_c$ intersects $\vec{t}$, it can be bended by a sufficiently small amount at cost at most $\lambda$ (the parameter that our solution depends on) so that it does not intersect $\vec{t}$ or any other of the straight line segments. For $\lambda \to 0$, the cost of this solution, i.e., $\sum_{\{\vec{p}_1, \vec{p}_2\} \in M_c} \lVert \vec{p}_1-\vec{p}_2 \rVert$, gives the desired entry of the lookup-table.
Therefore, assume from now on additionally that there are subportals colored $c^*$. 

For one of the matched pairs in $M_{c^*}$, the curve connecting the two subportals has to visit the terminal $\vec{t}$.  
Say this pair is given by $\{\vec{p}_1, \vec{p}_2\}$. Then, we obtain a solution as follows: First, include the curve $\overline{\vec{p}_1\vec{t}} \cup \overline{\vec{t}\vec{p}_2}$ in $\Pi_{c^*}$. Then, add the straight line segments  $\overline{\vec{q}_1\vec{q}_2}$ to $\Pi_c$ for all $\{\vec{q}_1,\vec{q}_2\} \in   M_c$ where $\overline{\vec{q}_1\vec{q}_2}$ does not intersect $\overline{\vec{p}_1\vec{t}} \cup \overline{\vec{t}\vec{p}_2}$ ($c \in C$ and $\{\vec{q}_1, \vec{q}_2\} \neq \{\vec{p}_1, \vec{p}_2\}$).
For each remaining pair $\{\vec{q}_1,\vec{q}_2\} \in   M_c$ that intersects $\overline{\vec{p}_1\vec{t}} \cup \overline{\vec{t}\vec{p}_2}$, there is a suitable point $\vec{t}'$  at distance at most $\lambda$ from $\vec{t}$ such that $\overline{\vec{q}_1\vec{t}'} \cup \overline{\vec{t}'\vec{q}_2}$ does not cross any other curves. 
This is illustrated by one of the green pairs on the left side in Figure~\ref{fig:multipath}.
We add $\pi:=\overline{\vec{q}_1\vec{t}'} \cup \overline{\vec{t}'\vec{q}_2}$ to $\Pi_c$. Note that $\l(\pi)=\lVert \vec{q}_1 - \vec{t'}\rVert + \lVert \vec{q}_2 - \vec{t'}\rVert = \lVert \vec{q}_1 - \vec{t}\rVert + \lVert \vec{q}_2 - \vec{t}\rVert$ for $\lambda \to 0$. Therefore, given a choice of the pair  $\{\vec{p}_1, \vec{p}_2\}$, we can efficiently compute the optimal cost of a solution.  
To find the entry of the lookup-table, we compute the cost for every choice of the pair $\{\vec{p}_1, \vec{p}_2\} \in \bigcup_{c\in C} M_c$ and check which choice minimizes the cost.%

To summarize, there are $O(\crossingnumber r \log L)$ subportals so that there are at most $O(\crossingnumber^2 r^2 \log^2 L)$ possible choices for the pair $\{\vec{p}_1, \vec{p}_2\}$. Once a pair is chosen, computing the cost of a solution (as $\lambda \to 0$) can be done in time $O(\crossingnumber r \log L)$ because the number of curves is bounded by the number of subportals.
Therefore, a single entry of the lookup-table of a leaf can be computed in time
$O(\crossingnumber^2r^2 \log^2 L)\leq L^{O(\crossingnumber r)}$.
Combining this with the fact that the size of every lookup-table is at most $L^{O(\crossingnumber r \log k)}$ (cf.~Observation~\ref{obs:lookup-table-size}), we obtain that $\lt$ can be computed in time $L^{O(\crossingnumber r \log k)}$.
\end{proof}
}
\movetoappendix{\prooflemlookuptableleaf}
\includeinappendix{
\begin{proof}[Proof sketch]
Observe on the left side of Figure~\ref{fig:multipath} that, for one of the red pairs of subportals in the matching, the curve connecting them has to visit the terminal. Once such a pair is chosen, it is straightforward how the other matched subportals can be connected optimally.  
\end{proof}
}

Next, we investigate how the lookup-table of a non-leaf node of $D(\vec{a})$ can be computed. 

\begin{restatable}{lemma}{lemlookuptablenonleaf}\label{lem:lookup-table-non-leaf}
For every non-leaf node $S$ of $D(\vec{a})$, its lookup-table $\lt$ can be computed in time
$L^{O(\crossingnumber r \log k)}$ if the lookup-tables of its four children are given.
\end{restatable}

\def\prooflemlookuptablenonleaf{
\begin{proof}
Let a square $S$ of side length $>1$, a coloring of the subportals $(P_c)_{c\in C}$ and a \noncrossing matching $(M_c)_{c \in C}$ be given. Let $S^{(1)}, S^{(2)}, S^{(3)}, S^{(4)}$ denote its four children in $D(\vec{a)}$ and assume that the lookup-tables of the children are already computed. 

Note that each child $S^{(i)}$ ($i \in \{1, \dots, 4\}$) shares two border edges with siblings and the other two of its border edges are contained in border edges of $S$ (cf.~right side of Figure~\ref{fig:multipath}). In particular, each $S^{(i)}$ shares subportals with siblings and the parent.
More formally, let $P$ denote the set of subportals on the border of $S$ and $P^{(i)}$ denote the set of subportals on the border of $S^{(i)}$ ($i \in \{1,2,3,4\})$. Then, we have $P^{(i)}\subseteq \bigcup_{j\neq i} P^{(j)} \cup P$.
In order to be able to combine the curves of solutions of the four children, we also need the following notion:%
a \emph{matched sequence of subportals} is a sequence of subportals  $\vec{p}_1, \dots, \vec{p}_N$ such that $\{\vec{p}_j, \vec{p}_{j+1}\}\in \bigcup_{c \in C, i \in \{1,\dots,4\}} M_c^{(i)}$ for every $j\in \{1,\dots, N-1\}$. Note that such a sequence exists if and only if combining the four subsolutions gives a curve that connects $\vec{p}_1$ with~$\vec{p}_N$. Also note that, in a matched sequence, $\vec{p}_2, \dots, \vec{p}_{N-1}$ cannot lie in $P$.

Given a combination of an entry of each lookup-table of the four children $(P_c^{(i)})_{c\in C}$, $(M_c^{(i)})_{c \in C}$, $(i \in \{1,2,3,4\})$, we say that it is \emph{compatible} with $(P_c)_{c\in C},(M_c)_{c \in C}$ if
\begin{enumerate}[a)]
\item the colorings of the shared subportals coincide, i.e., for $i,j \in \{1,2,3,4\}, c\in C$ and $p \in P^{(i)} \cap P^{(j)}$, we have $\vec{p} \in P_c^{(i)}$ if and only if $\vec{p} \in P_c^{(j)}$ and, similarly, for  $\vec{p} \in P^{(i)} \cap P$, we have $\vec{p} \in P_c^{(i)}$ if and only if $\vec{p} \in P_c$,
\item For every $c \in C$ and every $\{\vec{p}, \vec{q}\}\in M_c$, there is a matched sequence of portals  $\vec{p}_1, \dots, \vec{p}_N$ with $\vec{p}_1=\vec{p}$ and $\vec{p}_N=\vec{q}$,%
\item if $P_c\neq \emptyset$, for every maximal matched sequence of portals $\vec{p}_1, \dots, \vec{p}_N$, we have $\vec{p}_1,\vec{p}_N \in P$,
\item if $P_c=\emptyset$, $\bigcup_{i\in\{1,2,3,4\}, \{\vec{p},\vec{q}\}\in M_c^{(i)}} \overline{\vec{p}\vec{q}}$ is a single cycle.
\end{enumerate}
Note that condition c) implies that the combination of subsolutions is acyclic and only consists of curves touching $\partial S$.
An example of a compatible combination and the resulting solution is illustrated on the right hand side of Figure~\ref{fig:multipath}.

We obtain that the desired entry of $\lt$ is given by the following equation: 
\begin{align*}
    \lt \left[ (P_c)_{c\in C},(M_c)_{c \in C} \right]
    &= \min \biggl\{  \sum\limits_{i=1}^4 \mathrm{LT}_{S^{(i)}} \hspace{-1mm}\left[ (P_c^{(i)})_{c\in C},(M_c^{(i)})_{c \in C} \right] :(P_c^{(i)})_{c\in C},(M_c^{(i)})_{c \in C}\\
    & \hspace{1.4cm}(i \in \{1,2,3,4\}) \text{ is compatible with } (P_c)_{c\in C},(M_c)_{c \in C} 
    \biggl\}.
\end{align*}

We investigate the complexity of computing this minimum:
There are $\left(L^{O(\crossingnumber r\log k)}\right)^4=L^{O(\crossingnumber r \log k)}$ combinations of entries of the lookup-tables of the four children.
Since there are $O(\crossingnumber r \log L)$ subportals in $S$ and its four children, checking whether a combination is compatible, takes time at most $O(\crossingnumber^2r^2\log^2L)\leq L^{O(\crossingnumber r \log k)}$, we obtain that $\lt$ can be computed in time $L^{O(\crossingnumber r \log k)}$.
\end{proof}
}
\movetoappendix{\prooflemlookuptablenonleaf}
\includeinappendix{
\begin{proof}[Proof sketch]
To compute a given entry of $\lt$, we check for every combination of entries of lookup-tables of its four children whether they are \emph{compatible} in the sense that solutions to them can be combined (cf.~right side of Figure~\ref{fig:multipath}). The cost of the resulting solution is then simply the sum of the four entries. We choose the compatible combination minimzing~this.
\end{proof}
}

\includeinappendix{Combining the facts that $D(\vec{a})$ has $O(L^2)$ nodes and that the lookup-table of every node can be computed efficiently (Lemmas~\ref{lem:lookup-table-leaf} and \ref{lem:lookup-table-non-leaf}) completes the proof of Theorem~\ref{thm:dynamic-programming}.}

\def\proofdynamicprogramming{
\begin{proof}[Proof of Theorem~\ref{thm:dynamic-programming}]
Observe that the value $\opt$ for the given instance of $\ktsp'$ is given by the solution to the multipath problem of $C(\vec{a})$ with $P_c=\emptyset$ for all $c \in C$.
Recall that $D(\vec{a})$ is a full 4-ary tree with 
$L^2$ leafs.
Therefore, $D(\vec{a})$ contains in total at most $O(L^2)
$ nodes.
As we have seen in Lemmas~\ref{lem:lookup-table-leaf}~and~\ref{lem:lookup-table-non-leaf}, the lookup-tables of the nodes of $D(\vec{a})$ can be computed using dynamic programming and computing a single lookup-table takes time at most 
$L^{O(\crossingnumber r\log k)}$.
Together, this gives that our algorithm has running time 
$O\left( L^2 \cdot L^{O(mr\log k)} \right) = L^{O(\crossingnumber r \log k)}$.
As usual, a parametrized solution $\Pi(\lambda)$ with $\lim_{\lambda \to 0} \cost (\Pi(\lambda))=\opt$ can be found by reverse-engineering the computation of the lookup-tables. 
\end{proof}
}

\movetoappendix{
Now, we have all the prerequisites in place to prove Theorem~\ref{thm:dynamic-programming}, i.e., we show that $\opt$ can be computed in time $L^{O(\crossingnumber r \log k)}$.
\proofdynamicprogramming
}

\section{Perturbation}\label{sec:perturbation}

In the previous section, we have focused on solving $\ttsp$ when the terminals have integer coordinates. In this section, we show how an input for general $\ttsp$ can be preprocessed such that we only have to solve an instance with terminals in $\Z^2$ and how a solution to the preprocessed instance can be transformed back into a solution to the original instance. For this, we use similar ideas as in \cite{arora}. %
Note that, if the terminal sets of different colors are in some sense far away from each other, we can find a tour for each color separately such that they are are disjoint. 
More precisely, we call an instance to $\ttsp$ \emph{$\epsilon$-reducible}, if there is a choice of the three colors $\{c,c',c''\}=\rgb$ such that, when applying the algorithm in~\cite{dross} on $T_c, T_{c'}$ (with the given $\epsilon$), and, independently appylying Arora's algorithm \cite{arora} on $T_{c''}$, the resulting tours are disjoint and, therefore, provide a $(1+\epsilon)$-approximation for \ttsp.
For this reason, we restrict ourselves to instances that are non-reducible. 

\begin{figure}
\centering
\begin{tikzpicture}[scale=0.7]
\foreach \x in {1,...,9}{
\draw[-] (1,\x) to (9,\x);
\draw[-] (\x,1) to (\x,9);
} 
\foreach \x in {1,4,7}{
\filldraw[-, green!60!black] (9,\x) to (9.2, 0.1+\x) to (9.2,-0.1+\x) to (9,\x);
\filldraw[-, blue] (9,\x+1) to (9.2, 0.1+\x+1) to (9.2,-0.1+\x+1) to (9,\x+1);
\filldraw[-, red] (9,\x+2) to (9.2, 0.1+\x+2) to (9.2,-0.1+\x+2) to (9,\x+2);

\filldraw[-, green!60!black] (\x,9) to (0.1+\x,9.2) to (-0.1+\x,9.2) to (\x,9);
\filldraw[-, blue] (\x+1,9) to (0.1+\x+1,9.2) to (-0.1+\x+1,9.2) to (\x+1,9);
\filldraw[-, red] (\x+2,9) to (0.1+\x+2,9.2) to (-0.1+\x+2,9.2) to (\x+2,9);
}
\node (A) at (1,5.5) [circle, fill, red, inner sep=0pt, minimum size=2mm] {};
\node (Ap) at (3,6) [rectangle, fill, red, inner sep=0pt, minimum size=2mm] {};
\draw[->, red, thick] (A) to (Ap);
\node (B) at (2.45,4.9) [circle, fill, green!60!black, inner sep=0pt, minimum size=2mm] {};
\node (Bp) at (1,4) [rectangle, fill, green!60!black, inner sep=0pt, minimum size=2mm] {};
\draw[->, green!60!black, thick] (B) to (Bp);
\node (C) at (3.05,3.5) [circle, fill, blue, inner sep=0pt, minimum size=2mm] {};
\node (Cp) at (2,2) [rectangle, fill, blue, inner sep=0pt, minimum size=2mm] {};
\draw[->, blue, thick] (C) to (Cp);
\node (D) at (3.4,6.8) [circle, fill, blue, inner sep=0pt, minimum size=2mm] {};
\node (Dp) at (2,8) [rectangle, fill, blue, inner sep=0pt, minimum size=2mm] {};
\draw[->, blue, thick] (D) to (Dp);
\node (E) at (9,5.8) [circle, fill, blue, inner sep=0pt, minimum size=2mm] {};
\node (Ep) at (8,5) [rectangle, fill, blue, inner sep=0pt, minimum size=2mm] {};
\draw[->, blue, thick] (E) to (Ep);
\node (F) at (5.5, 8.5) [circle, fill, blue, inner sep=0pt, minimum size=2mm] {};
\node (Fp) at (5,8) [rectangle, fill, blue, inner sep=0pt, minimum size=2mm] {};
\draw[->, blue, thick] (F) to (Fp);
\node (G) at (3.8, 4.2) [circle, fill, red, inner sep=0pt, minimum size=2mm] {};
\node (Gp) at (3,3) [rectangle, fill, red, inner sep=0pt, minimum size=2mm] {};
\draw[->, red, thick] (G) to (Gp);
\node (H) at (7.1,7.2) [circle, fill, red, inner sep=0pt, minimum size=2mm] {};
\node (Hp) at (6,6) [rectangle, fill, red, inner sep=0pt, minimum size=2mm] {};
\draw[->, red, thick] (H) to (Hp);
\node (I) at (2.5,1) [circle, fill, green!60!black, inner sep=0pt, minimum size=2mm] {};
\node (Ip) at (1,1) [rectangle, fill, green!60!black, inner sep=0pt, minimum size=2mm] {};
\draw[->, green!60!black, thick] (I) to (Ip);
\end{tikzpicture}
\hspace{4mm}
\begin{tikzpicture}[scale=0.7]
\foreach \x in {1,...,9}{
\draw[-] (1,\x) to (9,\x);
\draw[-] (\x,1) to (\x,9);
} 
\foreach \x in {1,4,7}{
\filldraw[-, green!60!black] (9,\x) to (9.2, 0.1+\x) to (9.2,-0.1+\x) to (9,\x);
\filldraw[-, blue] (9,\x+1) to (9.2, 0.1+\x+1) to (9.2,-0.1+\x+1) to (9,\x+1);
\filldraw[-, red] (9,\x+2) to (9.2, 0.1+\x+2) to (9.2,-0.1+\x+2) to (9,\x+2);

\filldraw[-, green!60!black] (\x,9) to (0.1+\x,9.2) to (-0.1+\x,9.2) to (\x,9);
\filldraw[-, blue] (\x+1,9) to (0.1+\x+1,9.2) to (-0.1+\x+1,9.2) to (\x+1,9);
\filldraw[-, red] (\x+2,9) to (0.1+\x+2,9.2) to (-0.1+\x+2,9.2) to (\x+2,9);
}
\node (A) at (1,5.5) [circle, fill, red, inner sep=0pt, minimum size=2mm] {};
\node (Ap) at (3,6) [rectangle, fill, red, inner sep=0pt, minimum size=2mm] {};
\node (B) at (2.45,4.9) [circle, fill, green!60!black, inner sep=0pt, minimum size=2mm] {};
\node (Bp) at (1,4) [rectangle, fill, green!60!black, inner sep=0pt, minimum size=2mm] {};
\node (C) at (3.05,3.5) [circle, fill, blue, inner sep=0pt, minimum size=2mm] {};
\node (Cp) at (2,2) [rectangle, fill, blue, inner sep=0pt, minimum size=2mm] {};
\node (D) at (3.4,6.8) [circle, fill, blue, inner sep=0pt, minimum size=2mm] {};
\node (Dp) at (2,8) [rectangle, fill, blue, inner sep=0pt, minimum size=2mm] {};
\node (E) at (9,5.8) [circle, fill, blue, inner sep=0pt, minimum size=2mm] {};
\node (Ep) at (8,5) [rectangle, fill, blue, inner sep=0pt, minimum size=2mm] {};
\node (F) at (5.5, 8.5) [circle, fill, blue, inner sep=0pt, minimum size=2mm] {};
\node (Fp) at (5,8) [rectangle, fill, blue, inner sep=0pt, minimum size=2mm] {};
\node (G) at (3.8, 4.2) [circle, fill, red, inner sep=0pt, minimum size=2mm] {};
\node (Gp) at (3,3) [rectangle, fill, red, inner sep=0pt, minimum size=2mm] {};
\node (H) at (7.1,7.2) [circle, fill, red, inner sep=0pt, minimum size=2mm] {};
\node (Hp) at (6,6) [rectangle, fill,red, inner sep=0pt, minimum size=2mm] {};
\node (I) at (2.5,1) [circle, fill, green!60!black, inner sep=0pt, minimum size=2mm] {};
\node (Ip) at (1,1) [rectangle, fill, green!60!black, inner sep=0pt, minimum size=2mm] {};
\draw[-, very thick, blue] (Dp) to (2,5.7);
\draw[-, very thick, blue, rounded corners] (2,5.7) to (0.77, 5.7) to (0.77, 5.28) to  (2, 5.3);
\draw[-, very thick, blue] (2,5.3) to (2,4.85);
\draw[-, very thick, blue, rounded corners] (2,4.85) to (2.55, 5.2) to (2.7,4.8) to  (2, 4.4);
\draw[-, very thick, blue] (2,4.4) to (2,3.6);
\draw[-, very thick, blue] (2, 3.6) to (3.1, 3.5);
\draw[-, very thick, blue] (2, 3.4) to (3.1, 3.5);
\draw[-, very thick, blue] (2,3.4) to (Cp);
\draw[-, very thick, blue] (Cp) to (Ep);
\draw[-, very thick, blue] (9, 5.8) to (8.05, 4.95);
\draw[-, very thick, blue] (9, 5.8) to (7.9, 5.05);
\draw[-, very thick, blue] (7.9,5.05) to (6.65, 6.35);
\draw[-, very thick, blue, rounded corners] (6.65, 6.35) to (7.4, 7.2) to (7.1,7.5) to (6.33, 6.67);
\draw[-, very thick, blue] (6.33,6.67) to (Fp);
\draw[, very thick, blue] (5.07, 7.95) to (5.5, 8.5);
\draw[, very thick, blue] (4.93, 8.05) to (5.5, 8.5);
\draw[, very thick, blue] (Fp) to (3.5,8) to (3.4,6.8) to (3.4, 6.8) to (3.3, 8) to (Dp);

\draw[-, very thick, green!60!black] (0.92,0.95) to (0.92,4.05) to (2.45, 4.9) to (2.45, 4.9) to (1.08, 3.9) to (1.1, 1.1) to (2.5, 1) to (0.92,0.95);
\draw[-, very thick, red] (Gp) to (3.9, 3.9) to (3.8, 4.2) to (4.1, 4.1) to (Hp);
\draw[-, very thick, red] (6.05, 5.95) to (7.1, 7.2) to (5.9, 6.05);
\draw[-, very thick, red] (Hp) to (Ap) to (3, 5.6) to (1, 5.5)  to (3, 5.4) to (3,3.7);
\draw[-, very thick, red, rounded corners] (3,3.7) to (3.25, 3.7) to (3.25, 3.3) to (3, 3.3);
\draw[-, very thick, red] (3,3.3) to (Gp);
\end{tikzpicture}
\caption{Illustration of the algorithms \textsc{Perturbation} (left side) and \textsc{Back-Perturbation} (right side): The terminals in instance $\mathcal{I}$ are illustrated as circles and the terminals of the constructed instance $\mathcal{I}'$ are illustrated as rectangles. The triangles indicate on which gridlines the terminals of each color can be snapped on.}
\label{fig:perturbation}
\end{figure}
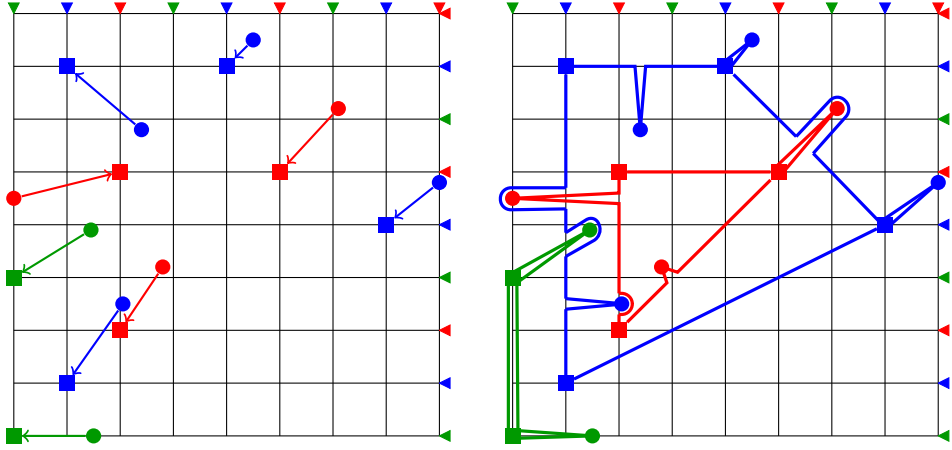

\begin{restatable}{theorem}{perturbation}\label{thm:perturbation}
Let $\epsilon>0$ and $\mathcal{I}$ be a non-$\epsilon$-reducible instance of $\ttsp$.
\begin{enumerate}[a)]
    \item There is an algorithm \textsc{Perturbation} that has running time $O(n^2)$ and returns an instance~$\mathcal{I}'$ of $\ktsp$ with terminals in $\L$ where $L=O(n/\epsilon)$ is a power of two.
    \item There is an algorithm \textsc{Back-Perturbation} that has running time $O(n^2)$ and, given a $(1+\epsilon')$-approximate solution to the instance $\mathcal{I}'$, returns a $(1+\epsilon'+O(\epsilon))$-approximate solution to the instance $\mathcal{I}$.
\end{enumerate}
\end{restatable}

\newcommand{\sidelength}{f}
\def\proofperturbation{
\begin{algorithm}[t]
\caption{\textsc{Perturbation}}\label{alg:perturbation}
\LinesNumbered
\SetKwInOut{Input}{input}
\Input{$\epsilon>0$, $T_\mathrm{R}, T_\mathrm{G}, T_\mathrm{B} \subseteq \R^2$ with $\bigl| \bigcup_{c \in \rgb} T_c \bigl|=:n$}
$(S,\sidelength) \gets \argmin \{\sidelength': S' \text{ is a square of side length } \sidelength \text{ such that } \bigcup_{c \in \rgb} T_c \subseteq S \}$\\
$L \gets \min \{2^i: 2^i \geq \lceil 3 \cdot 117 \cdot n /\epsilon \rceil, i \in \N \}$\\
let $v_0, \dots, v_{L}$ and $h_0, \dots, h_{L}$ denote $L+1$ vertical and $L+1$ horizontal equispaced straight line segments of length $\sidelength$ in $S$ with $h_0, v_0, h_{L}, v_{L} \subseteq \partial S'$\\
$T_{\mathrm{R}}', T_{\mathrm{G}}', T_{\mathrm{B}}' \gets \emptyset$\\
\ForAll{$c \in \rgb$ and $\vec{t} \in T_c$}{
$k \gets \begin{cases}
0, & \text{if } c=R\\
1, & \text{if } c=G\\
2, & \text{if } c=B\\
\end{cases}$\\
$\vec{t'} \gets \argmin \{\lVert \vec{t''}-\vec{t} \rVert : \vec{t''}=v_i \cap h_j \text{ with } i \hspace{-1.6mm}\mod 3 = j \hspace{-1.6mm} \mod 3 =k\}$\\
add $\vec{t'}$ to $T'_c$
}
identify the points $\{v_i \cap h_j : i,j \in \{0, \dots, L\} \}$ with $\L$\\
\Return{$L, T_{\mathrm{R}}', T_{\mathrm{G}}', T_{\mathrm{B}}'$}
\end{algorithm}

\begin{algorithm}[t]
\caption{\textsc{Back-Perturbation}}\label{alg:back-perturbation}
\LinesNumbered
\SetKwInOut{Input}{input}
\Input{$T_{\mathrm{R}},T_{\mathrm{G}},T_{\mathrm{B}}\subseteq [0,L]^2$, $(1+\epsilon')$-approximate solution $\pir', \pig', \pib'$ to an instance $T_{\mathrm{R}}', T_{\mathrm{G}}', T_{\mathrm{B}}' \subseteq \L$}
choose $\delta>0$ small enough\\
\ForAll{$c \in \rgb$ and $\vec{t} \in T_c$}{
$s\gets \argmin \{\l(s'):s' \text{ is a straight line segment between $\vec{t}$ and $\pi_c'$}\}$\\
\If{$s \cap \left( \left(\bigcup_{c\in \rgb} (T_c \cup T_c')\right)\setminus \{\vec{t}\})\right) \neq \emptyset$}{
replace $s$ by another segment $s'$ connecting $\vec{t}$ and $\pi_c$ such that  $s' \cap \left( \left(\bigcup_{c\in \rgb} (T_c \cup T_c')\right)\setminus \{\vec{t}\})\right) = \emptyset$ and $\l(s')\leq 2 \l(s)$
}
apply patching (Lemma~\ref{lem:patching}) to $\pi_{c'},\pi_{c''}$ along $s$ where $c',c''\neq c$\\
redirect the remaining crossings around $s$\\
choose $\vec{x}_1, \vec{x}_2 \in \pi_c'$ close enough at distance at most $\delta$ to $s \cap \pi_c$ such that $\pi_c[\vec{x}_1,\vec{x}_2] \cap T=\emptyset$\\
replace $\pi_c'$ by $(\pi_c \setminus \pi_c[\vec{x}_1,\vec{x}_2]) \cup \overline{\vec{x_1}\vec{t}} \cup \overline{\vec{x_2}\vec{t}}$
}
\Return{$\pir', \pig', \pib'$}
\end{algorithm}

\begin{proof}%
We begin with proving part a). For this, consider the algorithm \textsc{Perturbation} given in Algorithm~\ref{alg:perturbation}, which is illustrated on the left of Figure~\ref{fig:perturbation}.
We begin by investigating its running time. For this, observe that line 1 can be executed in time $O(n^2)$, the for-loop in line 5 is executed $n$ times, and all other steps can be executed in time $O(1)$. Therefore, its total running time is $O(n^2)$.

Next, note that the sets $(T_c')_{c\in\rgb}$ are disjoint because, due to lines 6 and 7, terminals of different colors cannot be placed on the same line $h_i$ or $v_j$. Therefore, the output is an instance of $\ktsp$. It is immediate from the construction that its terminals lie in $\L$ and that $L=O(n/\epsilon)$. This completes the proof of part a).

\looseness=-1
Before moving on to \textsc{Back-Perturbation}, we note some more useful properties of the algorithm \textsc{Perturbation}.
Observe that the assumption that the instance is non-$\epsilon$-reducible implies that  $f\leq O(\opt)$:
This is because, given a $(1+\epsilon)$-approximation to the instance $T_c, T_{c'}$ of $\twotsp$ and a $(1+\epsilon)$-approximation to the instance $T_{c''}$ of $\ttsp$ such that $\pi_{c''}$ and $\pi_c$ intersect, we obtain that the distance between any pair of terminals of colors $c''$ and $c$ is at most $2(1+\epsilon)\opt$. Using this for every choice of the three colors $\{c,c',c''\}=\rgb$, we obtain that the distance between any two terminals is at most $2(1+\epsilon)\opt$, so that $f\leq O(\opt)$.
Therefore, the distance between $h_i$ and $h_{i+1}$ (or $v_i$ and $v_{i+1}$) is at most
$\sidelength/L\leq O(\opt /L)$.
It follows that, for every terminal $\vec{t}$ in $T_c$, the distance to the closest terminal in $T_c'$ is at most 
$\sqrt{2}\cdot 3\cdot O(\opt/L)\leq O( \opt/L )$ (due to line 7, cf.~left side of Figure~\ref{fig:perturbation}).

Now, we turn to proving part b). For this investigate the algorithm \textsc{Back-Perturbation} given in Algorithm~\ref{alg:back-perturbation} and illustrated on the right side of Figure~\ref{fig:perturbation}.
First, note that every single line can be executed in time $O(n)$. Since the loop beginning in line 2 is executed $n$ times, this gives a running time of $O(n^2)$.
Next, note that, by construction, the resulting tours are indeed a solution to the instance of $\ktsp$ with terminal sets $T_{\mathrm{R}},T_{\mathrm{G}},T_{\mathrm{B}}$.

It remains to estimate the cost of the resulting tours. Hereby, we assume that 
$L,T_{\mathrm{R}}',T_{\mathrm{G}}',T_{\mathrm{B}}'$ are the output of 
\textsc{Perturbation} on $T_{\mathrm{R}},T_{\mathrm{G}},T_{\mathrm{B}}$. 
By $\opt$, we refer to the infimum cost of solutions to $\mathcal{I}$ and, by 
$\opt_p$, we refer to the infimum cost of solutions to the perturbed 
instance~$\mathcal{I}'$.
Note that the length of a segment $s$ chosen in the algorithm is at most  $O(\opt/L)$ because we have already seen that, for every terminal $t$ in $T_c$, the distance to the closest terminal in $T_c'$ is at most 
$O(\opt/L)$.%
Therefore, in line 6, the total lengths of the tours are increased at most by $O(\opt /L)$ (cf.~Lemma~\ref{lem:patching}).
Then, the remaining number of crossings is at most 18 so that the additional cost in line 7 is $O(\opt /L)$ as well. 
Last, in line 9, since $\delta$ can be chosen small enough, the length of the considered tour is also increased by at most $O(\opt /L)$.
Since the loop in line 2 is executed at most $n$ times, the total cost of the resulting tours is 
\begin{equation*}
(1+\epsilon')\cdot \opt_p + O\left(\frac{n \cdot \opt}{L}\right)=(1+\epsilon'+O(n/L))\cdot \opt_p=(1+\epsilon'+O(\epsilon))\cdot \opt_p, 
\end{equation*}
where we have used that $L=O(n/\epsilon)$.

It remains to relate $\opt$ to $\opt_p$.
For this, given a $(1+\eta)$-approximate solution to $\mathcal{I}$, note that we can 
use the algorithm \textsc{Back-perturbation} with exchanged roles of $\mathcal{I}$ and $\mathcal{I}'$
to construct a solution to $\mathcal{I'}$ of length at most $(1+\eta+ O(\epsilon))\cdot \opt$.
Letting $\eta\to 0$, this proves that $\opt_p\leq (1+O(\epsilon))\cdot \opt$.

With this, we obtain that the cost of the constructed solution is at most
\begin{equation*}
(1+\epsilon'+O(\epsilon))\cdot \opt_p
\leq (1+\epsilon'+O(\epsilon))\cdot (1+O(\epsilon))\cdot \opt
\leq (1+\epsilon'+O(\epsilon))\cdot \opt_,
\end{equation*}
which completes the proof.
\end{proof}
}
\movetoappendix{\proofperturbation}
\includeinappendix{
\begin{proof}[Proof sketch]
The main idea for the algorithm $\textsc{Perturbation}$ is the following (cf.~left side of Figure~\ref{fig:perturbation}): First, choose a square $S$ of minimum size containing all terminals. The non-reducibility condition allows us to derive that the side length of $S$ is at most $O(\opt)$. Then, place a grid in $S$ with $O(n/\epsilon)$ many grid lines. 
The algorithm snaps each terminal to a close enough intersection point of these grid lines without moving two terminals of different colors to the same point.

For \textsc{Back-Perturbation} (cf.~right side of Figure~\ref{fig:perturbation}), include a straight line segment between terminals of the original instance and the solution to the perturbed instance. By definition of $\textsc{Perturbation}$, such a segments is not too long (of length $O(\epsilon\opt/n$)), but it might cross the two other tours. Therefore, we apply patching to the other two tours along the segment such that the number of crossings is at most a constant. Then, replace each crossing by a detour around the segment.
\end{proof}
}

\section{A $( \frac{5}{3}+\epsilon )$-approximation algorithm for $\ttsp$}

We have all the prerequisites in place to prove our main result. We begin by recalling the theorem. 

\setcounter{theorem}{0}
\begin{restatable}[restated]{theorem}{mainthm}
    For every $\epsilon>0$, there is an algorithm that computes a $\left( \frac{5}{3}+\epsilon \right)$-approximation for $\ttsp$ in time $(\frac{n}{\epsilon})^{O(1/\epsilon)}$.
\end{restatable}

\def\algmainthm{
\newcommand{\sol}{\text{Sol}}
\begin{algorithm}[t]
\caption{$\left( \frac{5}{3}+\epsilon \right)$-approximation for $\ttsp$}\label{alg:5-3-approximation}
\LinesNumbered
\SetKwInOut{Input}{input}
\SetKwInOut{Parameters}{parameters}
\Parameters{ large enough constant $M$}
\Input{$\epsilon>0$, disjoint terminal sets $T_\mathrm{R}, T_\mathrm{G}, T_\mathrm{B} \subseteq \L$}
$(L, T_\mathrm{R}', T_\mathrm{G}', T_\mathrm{B}') \gets \textsc{Perturbation}(\epsilon/M, T_\mathrm{R}, T_\mathrm{G}, T_\mathrm{B})$\\
choose $\delta, \mu>0$ small enough\\
$\sol \gets \emptyset$\\
\ForAll{$\vec{a}\in \Lmin$}{
compute a $(1+\mu)$-approximate solution $\Pi=(\pir,\pig,\pib)$ for $\ttsp'$ with  terminals $T_\mathrm{R}', T_\mathrm{G}', T_\mathrm{B}'$ that is $(\lceil M(15\sqrt{2}+4)/\epsilon \rceil, 18, \delta)$-portal-respecting in $D(\vec{a})$\\
$\sol \gets \sol \cup \{\Pi\}$\\
\textcolor{gray}{
\ForAll{$\{c,c',c''\} = \{\mathrm{R},\mathrm{G},\mathrm{B}\}$}{
compute a $(1+\mu)$-approximate solution $(\pi_1,\pi_2)$  for $\twotsp'$ with terminals $T_{c^\ast}:=T_c \cup T_{c'}, T_{c''}$ and weights $w_{c^\ast}=2, w_{c''}=1$ such that the solution is $(\lceil M(15\sqrt{2}+4)/\epsilon \rceil, 36, \delta)$-portal-respecting in $D(\vec{a})$\\
transform $(\pi_1,\pi_2)$ into an induced two-tour solution $\Pi$\\
$\sol \gets \sol \cup \{\Pi\}$
}
}
}
$\Pi' \gets \argmin\{\cost(\Pi): \Pi \in \sol \}$\\
\Return{\emph{$\textsc{Back-Perturbation}(\Pi')$}}
\end{algorithm}
}

\def\proofmainthm{
\begin{proof}

We can check whether a given instance is $\epsilon$-reducible in time  $(\frac{n}{\epsilon})^{O(1/\epsilon)}$ by simply checking for all three choice of the colors $\{c,c',c''\}=\rgb$ such that, whether applying the algorithm in~\cite{dross} on $T_c, T_{c'}$ (with the given $\epsilon$), and, independently, Arora's algorithm \cite{arora} on $T_{c''}$, the resulting tours are disjoint.  In that case, we find a solution as desired. Therefore, assume that we are given a non-$\epsilon$-reducible instance.
In that case, we show that Algorithm~\ref{alg:5-3-approximation} computes a $\left( \frac{5}{3}+\epsilon \right)$-approximation for $\ttsp$ in time $(\frac{n}{\epsilon})^{O(1/\epsilon)}$.

In Theorem~\ref{thm:perturbation} a), we have seen how to perturb the input instance $\mathcal{I}$ to $\ttsp$ to obtain a perturbed instance $\mathcal{I}':=(T_\mathrm{R}', T_\mathrm{G}', T_\mathrm{B}')$ to $\ttsp$ with terminals 
in $\L$ for $L=O(n/\epsilon)$.
Let $\opt_p$ denote the optimal cost for the instance $\mathcal{I}'$.

Let $M$ be a (large) parameter to be chosen later. By applying Theorem~\ref{thm:structure-thm} on $\mathcal{I}'$ and $\epsilon/M$, we obtain that there exists 
a shift vector $\vec{a}\in\Lmin$ and $\delta>0$ such that there is either a $(\lceil M(15\sqrt{2}+4)/\epsilon \rceil, 18, \delta)$-portal respecting solution of cost at most $(1+\epsilon/M)\cdot \opt_p$, or there is a $(\lceil M(15\sqrt{2}+4)/\epsilon \rceil, 18, \delta)$-portal respecting \twotoursol $(\pi_1,\pi_2)$ with $2\l(\pi_1)+\l(\pi_2)\leq \left(\frac{5}{3}+\frac{\epsilon}{2M} \right)\cdot \opt_p$.
In the first case, one of the solutions computed in line 5 has cost at most $(1+\epsilon/M)\cdot \opt_p$.
In the other case, one of the \twotoursols computed in line 8 fulfills $2\l(\pi_1)+\l(\pi_2)\leq \left(\frac{5}{3}+\frac{\epsilon}{2M} \right)\cdot \opt_p$. Therefore, in line 9, it can be transformed into an induced two-tour solution of cost at most $\left(\frac{5}{3}+\frac{\epsilon}{M} \right)\cdot \opt_p$.
In either case, we obtain that $\Pi'$ chosen in line 11 is a solution to the instance $\mathcal{I}'$ of cost at most 
$\left(\frac{5}{3}+\frac{\epsilon}{M} \right)\cdot \opt_p$.

By Theorem~\ref{thm:perturbation} a), $\textsc{Back-Perturbation}(\Pi')$ is a solution to the instance $\mathcal{I}$ of cost at most
$\left( 1+\frac{5}{3}+\frac{\epsilon}{M} + O\left(\frac{\epsilon}{M}\right)\right) \cdot \opt$.
Therefore, it is possible to choose the constant $M$ such that the cost of $\Pi'$ is at most $(1+\epsilon)\cdot \opt$ as desired.

It remains to estimate the running time of Algorithm~\ref{alg:5-3-approximation}.
By Theorem~\ref{thm:perturbation}, the steps in line~1 and 12 have running time 
$O(n^2)$. By Theorem~\ref{thm:dynamic-programming}, the steps in line 5 and 8 
have running time $\smash{L^{O((M/\epsilon)\cdot 18 \cdot 
\log(3))}=\left(\frac{n}{\epsilon}\right)^{O(1/\epsilon)}}$. Since the for-loop 
is executed $L^2=O\left( \left( \frac{n}{\epsilon}\right)^2 \right)$ times, the 
overall running time is 
$\smash{\left(\frac{n}{\epsilon}\right)^{O(1/\epsilon)}}$, which completes the 
proof of the Theorem.
\end{proof}
Observe that, due to Corollary~\ref{cor:structure-theorem}, Algorithm~\ref{alg:5-3-approximation} can be shortened by deleting lines 7-10 and increasing the allowed number of crossings per portal to 36. 
However, our algorithm has the nice property that the minimal-cost solution computed in lines 7-10 gives us a $(1+\epsilon)$-approximation for the problem of finding an induced two-tour solution minimizing the cost.
}
\movetoappendix{\algmainthm
\proofmainthm}

\includeinappendix{
\begin{proof}[Proof sketch]
It is straightforward how to check whether an instance is $\epsilon$-reducible and how we can find a solution in that case so we assume that we are given a non-reducible instance.
Then, we first apply algorithm $\textsc{Perturbation}$ to the instance. By Corollary~\ref{cor:structure-theorem}, there is a portal-respecting solution to the perturbed instance $\mathcal{I}'$ that gives a $\left( \frac{5}{3}+\epsilon' \right)$-approximation (where we need to choose $\epsilon'$ to be $\epsilon$ divided by a large enough constant). 
Therefore, we can apply Theorem~\ref{thm:dynamic-programming} to compute a $\left( \frac{5}{3}+\epsilon' \right)$-approximate portal-respecting solution to $\cal I'$. Last, we apply $\textsc{Back-Perturbation}$ to the solution.
\end{proof}
}

\bibliography{3TSP}

\includeinappendix{
\newpage
\appendix

\section{Deferred proofs from Section 2}

\lemtwotoursol*
\prooftwotoursol

\section{Deferred proofs from Section 3}
\lemportals*
\prooflemportals

\patching*
\proofpatching

\structuretheorem*
\proofstructuretheorem

\section{Deferred proofs from Section 4}
\lemnoncrossingmatchings*
\prooflemnoncrossingmatchings

\lemlookuptableleaf*
\prooflemlookuptableleaf

\lemlookuptablenonleaf*
\prooflemlookuptablenonleaf

With this, we have the prerequisite in place to prove Theorem~\ref{thm:dynamic-programming}.

\thmdynamicprogramming*
\proofdynamicprogramming

\section{Deferred proofs from Section 5}
\perturbation*
\proofperturbation

\section{Deferred proofs from Section 6}
\mainthm*
\algmainthm
\proofmainthm

}

\end{document}